\documentclass[prodmode,acmtocl]{acmsmall}

\usepackage[ruled]{algorithm2e}

\SetAlFnt{\small}
\SetAlCapFnt{\small}
\SetAlCapNameFnt{\small}
\SetAlCapHSkip{0pt}
\IncMargin{-\parindent}

\acmVolume{V}
\acmNumber{N}
\acmArticle{1}
\acmYear{20YY}
\acmMonth{1}

\usepackage{latexsym}
\usepackage{amssymb}
\usepackage{amsmath}
\usepackage{stmaryrd}
\usepackage{enumerate}

\newtheorem{claim}{Claim}
\newcommand{\nn}{\mbox{${\mathbb N}$}}

\newcommand{\Domain}{\mbox{$\textsf{Dom}$}}
\newcommand{\val}{\mbox{$\textrm{\em va}\ell$}}

\newcommand{\lab}{\mbox{$\ell ab$}}

\newcommand{\eda}{\mbox{$E_{\downarrow}$}}
\newcommand{\era}{\mbox{$E_{\rightarrow}$}}

\renewcommand{\root}{\mbox{{\rm root}}}
\newcommand{\ls}{\mbox{{\rm last-sibling}}}
\newcommand{\fs}{\mbox{{\rm first-sibling}}}
\newcommand{\leaf}{\mbox{{\rm leaf}}}

\newcommand{\A}{\mbox{$\mathcal{A}$}}

\newcommand{\C}{\mbox{$\mathcal{C}$}}
\newcommand{\D}{\mbox{$\mathcal{D}$}}
\newcommand{\F}{\mbox{$\mathcal{F}$}}
\renewcommand{\L}{\mbox{$\mathcal{L}$}}
\newcommand{\M}{\mbox{$\mathcal{M}$}}
\newcommand{\N}{\mbox{$\mathcal{N}$}}
\renewcommand{\O}{\mbox{$\mathcal{O}$}}
\renewcommand{\P}{\mbox{$\mathcal{P}$}}

\renewcommand{\S}{\mbox{$\mathcal{S}$}}
\newcommand{\T}{\mbox{$\mathcal{T}$}}

\newcommand{\V}{\mbox{$\mathcal{V}$}}

\newcommand{\Val}{\mbox{$\textrm{Val}$}}

\newcommand{\sA}{\mbox{\scriptsize $\mathcal{A}$}}

\newcommand{\sM}{\mbox{\scriptsize $\mathcal{M}$}}
\newcommand{\sP}{\mbox{\scriptsize $\mathcal{P}$}}
\newcommand{\sT}{\mbox{\scriptsize $\mathcal{T}$}}

\newcommand{\MSO}{\mbox{$\textrm{MSO}$}}
\newcommand{\FO}{\mbox{$\textrm{FO}$}}
\newcommand{\EMSO}{\mbox{$\exists\MSO$}}

\newcommand{\frU}{\mbox{$\mathfrak{U}$}}

\newcommand{\sfrU}{\mbox{\scriptsize $\mathfrak{U}$}}

\newcommand{\bbN}{\mbox{$\mathbb{N}$}}
\newcommand{\bbS}{\mbox{$\mathbb{S}$}}
\newcommand{\sbbS}{\mbox{\scriptsize $\mathbb{S}$}}

\newcommand{\np}{\mbox{$\textsc{NP}$}}
\newcommand{\nexp}{\mbox{$\textsc{NExpTime}$}}
\newcommand{\exptime}{\mbox{$\textsc{ExpTime}$}}

\newcommand{\sfMSOqr}{\mbox{$\textsf{MSO-qr}$}}
\newcommand{\sfFOqr}{\mbox{$\textsf{FO-qr}$}}

\newcommand{\scroot}{\mbox{$\textsc{root}$}}

\newcommand{\dvSucc}{\mbox{$\prec_{suc}$}}

\newcommand{\Parikh}{\mbox{$\textsf{Parikh}$}}

\newcommand{\sfProfile}{\mbox{$\textsf{Profile}$}}

\newcommand{\sfProj}{\mbox{$\textsf{Proj}$}}

\newcommand{\sem}[1]{\llbracket #1 \rrbracket}
\newcommand{\semt}[1]{{\sem #1}_t}

\newcommand{\LTL}{\mbox{$\textrm{LTL}$}}
\newcommand{\ttU}{\mbox{$\texttt{U}$}}
\newcommand{\ttX}{\mbox{$\texttt{X}$}}

\newcommand{\RA}{\mbox{$\textrm{RA}$}}

\newcommand{\OMIT}[1]{}

\title{Extending two-variable logic on data trees
with order on data values and its automata}
            
\author{Tony Tan\\
Hasselt University and Transnational University of Limburg}
            
\begin{abstract} 
Data trees are trees in which each node, besides carrying a label from a finite alphabet, 
also carries a data value from an infinite domain. 
They have been used as an abstraction model for reasoning tasks on {XML} and verification. 
However, most existing approaches consider the case where 
only equality test can be performed on the data values.

In this paper we study data trees in which the data values come from a linearly ordered domain, 
and in addition to equality test, 
we can test whether the data value in a node is greater than the one in another node. 
We introduce an automata model for them which 
we call {\em ordered-data tree automata} (ODTA), provide its logical characterisation,
and prove that its non-emptiness problem is decidable in 3-$\nexp$. 
We also show that the two-variable logic on unranked data trees, 
studied by Bojanczyk, Muscholl, Schwentick and Segoufin in 2009, corresponds precisely to
a special subclass of this automata model.

Then we define a slightly weaker version of ODTA, which we call {\em weak ODTA}, 
and provide its logical characterisation. 
The complexity of the non-emptiness problem drops to $\np$.
However, a number of existing formalisms and models studied in the literature 
can be captured already by weak ODTA.
We also show that the definition of ODTA can be easily modified,
to the case where the data values come from a tree-like partially ordered domain, 
such as strings.
\end{abstract}
            
         
            

\begin{document}
            
\begin{bottomstuff} 
The extended abstract of this article has been published in
the proceedings of LICS 2012, 
under the title: ``An Automata Model for Trees with Ordered Data Values.''
The author is FWO Pegasus Marie Curie Fellow.
\end{bottomstuff}
            
\maketitle

\section{Introduction}
\label{s: intro}

Classical automata theory studies words and trees over finite alphabets.
Recently there has been a growing interest in the so-called ``data'' words and trees,
that is, words and trees in which each position,
besides carrying a label from a finite alphabet,
also carries a data value from an infinite domain.

Interest in such structures with data springs due to their connection to
XML~\cite{tova-typing-one,AFL-sicomp,dortmund-mfcs08,edt,FL-jacm,figueira-pods09,neven-csl},
as well as system 
specifications~\cite{therien-and-co,demri-souza-gascon-lfcs07,segoufin-torunczyk-stacs11},
where many properties simply cannot be captured by finite alphabets.
This has motivated various works on 
data words~\cite{benedikt-ley-puppis-csl10,BDMSS11,demri-lazic-tocl,grumber-kupferman-sheinvald-lata10,kaminski-francez,NSV04},
as well as on 
data trees~\cite{henrik-mikolaj,BMSS09,Fig12b,figueira-segoufin-stacs11,lazic-lics07}.
The common feature of these works is 
the addition of equality test on the data values to the logic on trees.
While for finitely-labeled trees many
logical formalisms (e.g., the monadic second-order logic $\MSO$) 
are decidable by converting formulae to automata,
even $\FO$ (first-order logic) on data words extended with data-equality
is already undecidable.
See, e.g.,~\cite{BDMSS11,FL-jacm,NSV04}.

Thus, there is a need for expressive enough, while computationally well-behaved, 
frameworks to reason about structures with data values. 
This has been quite a common theme in XML and 
system specification research.
It has largely followed two routes. 
The first takes a specific reasoning task, or a set of similar tasks, 
and builds algorithms for them (see,
e.g.,~\cite{AFL-sicomp,Fig11,dortmund-mfcs08,schwentick-sigmodr,FL-jacm,figueira-pods09}).
The second looks for sufficiently general automata models
that can express reasoning tasks of interest, but are still decidable
(see, e.g., \cite{demri-lazic-tocl,BMSS09,lazic-lics07,segoufin-torunczyk-stacs11}).

Both approaches usually assume that data values 
come from an abstract set equipped only with the equality predicate.
This is already sufficient to capture a wide range of interesting applications
both in databases and verification.
However, it has been advocated in~\cite{deutsch-HPV-data-icdt09}
that comparisons based on a linear order over the data values could be useful
in many scenarios, including data centric applications built on top of a database.

So far, not many works have been done in this direction.
A few works such as~\cite{fo2-amal,Fig11,schwentick-zeume,segoufin-torunczyk-stacs11}
are on words, while in most applications we need to consider trees.
Moreover, these works are incomparable to some interesting existing formalisms
\cite{FL-jacm,BMSS09,AFL-sicomp,edt,lazic-lics07,demri-lazic-tocl,Lazic11}
known to be able to capture various interesting scenarios common in practice.
On top of that many useful techniques,
notably those introduced in~\cite{FL-jacm,BDMSS11,BMSS09,lazic-lics07},
can deal only with data equality, and are highly dependent
on specific combinatorial properties of the formalisms.
They are rather hard to adapt to other more specific tasks,
let alone being generalised to include more relations on data values,
and they tend to produce extremely high complexity bounds,
such as non-primitive-recursive, or 
at least as hard as the reachability problem in Petri nets.
Furthermore, many known decidability results are lost
as soon as we add the order relation on data values.
Some exceptions are~\cite{FHL10,Fig12b}.

In this paper we study the notion of data trees in which the data values
come from a linearly ordered domain, which we call {\em ordered-data trees}.
In addition to equality tests on the data values,
in ordered-data trees we are allowed to test whether 
the data value in a node is greater than the data value in another node.
To the extent it is possible,
we aim to unify various ad hoc methods introduced to reason about data trees,
and generalise them to ordered-data trees to
make them more accessible and applicable in practice.
This paper is the first step,
where we introduce an automata model for ordered-data trees,
provide its logical characterisation, and prove that it has decidable non-emptiness problem.
Moreover, we also show that it can capture various well known formalisms.

\paragraph*{Brief description of the results in this paper}
The trees that we consider are {\em unranked} trees where
there is no a priori bound in the number of children of a node.
Moreover, we also have an order on the children of each node.
We consider a natural logic for ordered-data trees,
which consists of the following relations.
\begin{itemize}\itemsep=0pt
\item
The parent relation $\eda$, where
$\eda(x,y)$ means that node $x$ is the parent of node $y$. 
\item
The next-sibling relation $\era$, where 
$\era(x,y)$ means that nodes $x$ and $y$ have the same parent
and $y$ is the next sibling of $x$.
\item
The labeling predicates $a(\cdot)$'s,
where $a(x)$ means that node $x$ is labeled with symbol $a$.
\item
The data equality predicate $\sim$,
where $x\sim y$ means that nodes $x$ and $y$ have the same data value.
\item
The order relation on data $\prec$,
where $x \prec y$ means that the data value in node $x$ is less than
the one in node $y$.
\item
The successive order relation on data $\dvSucc$,
where $x \dvSucc\ y$ means that the data value in node $y$
is the minimal data value in the tree greater than the one in node $x$.
\end{itemize}

We introduce an automata model for ordered-data trees, 
which we call {\em ordered-data tree automata} (ODTA), 
and provide its logical characterisation.
Namely, we prove that the class of languages accepted by ODTA
corresponds precisely to those expressible by formulas of the form:
\begin{equation}
\label{eq: formula}
\exists X_1 \cdots \exists X_n 
\ \varphi \wedge \psi,
\end{equation}
where 
\begin{itemize}\itemsep=0pt
\item
$X_1,\ldots,X_n$ are monadic second-order predicates;
\item
$\varphi$ is an FO formula restricted to two variables
and using only the predicates $\eda$, $\era$, $\sim$,
as well as the unary predicates $X_1,\ldots,X_n$ and $a$'s.
\item
$\psi$ is an FO formula using only the predicates $\sim$, $\prec$, $\dvSucc$,
as well as the unary predicates $X_1,\ldots,X_n$ and $a$'s.
\end{itemize}
We show that the logic $\EMSO^2(\eda,\era,\sim)$,
first studied in~\cite{BMSS09}, corresponds precisely
to a special subclass of ODTA,
where $\EMSO^2(\eda,\era,\sim)$ denotes
the set of formulas of the form~(\ref{eq: formula})
in which $\psi$ is a true formula.
We then prove that the non-emptiness problem of ODTA 
is decidable in 3-$\nexp$.
Our main idea here is to show how to convert
the ordered-data trees back to a string over {\em finite} alphabets.
(See our notion of {\em string representation of data values} in Section~\ref{s: ordered-data tree}.)
Such conversion enables us to use
the classical finite state automata to reason about data values.

Then we define a slightly weaker version of ODTA,
which we call {\em weak ODTA}.
Essentially the only feature of ODTA missing in weak ODTA
is the ability to test whether two adjacent nodes have the same data value.
Without such simple feature, the complexity of the non-emptiness problem
surprisingly drops three-fold exponentially to $\np$.
We provide its logical characterisation by showing that
it corresponds precisely to the languages 
expressible by the formulas of the form~(\ref{eq: formula})
where $\varphi$ does not use the predicate $\sim$.
We show that a number of existing formalisms and models 
can be captured already by weak ODTA, i.e.
those in~\cite{FL-jacm,edt,fo2-amal}.

We should remark that~\cite{edt} studies a formalism
which consists of tree automata and a collection of
{\em set} and {\em linear} constraints.\footnote{We will later
define formally what set and linear constraints are.}
It is shown that the satisfiability problem of such formalism 
is $\np$-complete. In fact, it is also shown in~\cite{edt} that
a single set constraint (without tree automaton and linear constraint)
already yields $\np$-hardness.
Weak ODTA are essentially
equivalent to the formalism in~\cite{edt} extended 
with the full expressive power of the first-order logic $\FO(\sim,\prec,\dvSucc)$.
It is worth to note that despite such extension,
the non-emptiness problem remains in $\np$.

Finally we also show that the definition of ODTA can be easily modified
to the case where the data values come from a partially ordered domain, such as strings.
This work can be seen as a generalisation of the works 
in~\cite{fo2-lpar} and~\cite{fo2-lata12}.
However, it must be noted that~\cite{fo2-lpar,fo2-lata12}
deal only with {\em data words}, where
only equality test is allowed on the data values
and there is no order on them.

\paragraph*{Related works}
Most of the existing works in this area are on data words.
In the paper~\cite{BDMSS11} the model {\em data automata} was introduced, and 
it was shown that it captures the logic $\exists\MSO^2(\sim,<,+1)$,
the fragment of existential monadic second order logic in which
the first order part uses only two variables
and the predicates: the data equality $\sim$,
as well as the order $<$ and the successor $+1$ on the domain.

An important feature of data automata is
that their non-emptiness problem is decidable, even for infinite words,
but is at least as hard as reachability for Petri nets.
It was also shown that the satisfiability problem for
the three-variable first order logic is undecidable.
Later in~\cite{fo2-lpar} an alternative proof was given
for the decidability of the weaker logic $\EMSO^2(+1,\sim)$.
The proof gives a decision procedure 
with an elementary upper bound for the satisfiability
problem of $\EMSO^2(+1,\sim)$ on strings.
Recently in~\cite{fo2-lata12} an automata model that captures precisely
the logic $\EMSO^2(+1,\sim)$, both on finite and infinite words, is proposed.

Another logical approach is via the so called
{\em linear temporal logic} with freeze quantifier,
introduced in~\cite{demri-lazic-tocl}.
Intuitively, these are LTL formulas equipped
with a finite number of registers to store the data values.
We denote by $\LTL_n^{\downarrow}[\texttt{X},\texttt{U}]$,
the LTL with freeze quantifier,
where $n$ denotes the number of registers and
the only temporal operators allowed are 
the neXt operator $\ttX$ and the Until operator $\ttU$.
It was shown that alternating register automata with $n$ registers ($\RA_n$)
accept all LTL$_n^{\downarrow}[\texttt{X},\texttt{U}]$ languages
and the non-emptiness problem for alternating $\RA_1$ is decidable.
However, the complexity is non primitive recursive.
Hence, the satisfiability problem for
LTL$_1^{\downarrow}(\texttt{X},\texttt{U})$
is decidable as well.
Adding one more register or past time operator $\texttt{U}^{-1}$
to LTL$_1^{\downarrow}(\texttt{X},\texttt{U})$
makes the satisfiability problem undecidable.
In~\cite{FHL10,Fig12b} it is shown that alternating $\RA_1$
can be extended to strings with linearly ordered data values,
and the emptiness problem is still decidable.
In~\cite{Lazic11} a weaker version of alternating $\RA_1$,
called safety alternating $\RA_1$, is considered,
and the non-emptiness problem is shown to be EXPSPACE-complete.

A model for data words with linearly ordered data values
was proposed in~\cite{segoufin-torunczyk-stacs11}.
The model consists of an automaton equipped with a finite number of registers,
and its transitions are based on constraints on the data values
stored in the registers.
It is shown that the non-emptiness problem for this model
is decidable in PSPACE.
However, no logical characterisation is provided
for such model.

In~\cite{aut-group-lics11} another type of register automata for words
was introduced and studied, which is a generalisation
of the original register automata introduced by Kaminski and Francez~\cite{kaminski-francez},
where the data values also can come from a linearly ordered domain.
Thus, the order comparison, not just equality,
can be performed on data values.
The generalisation is done via the notion of monoid for data words,
and is incomparable with our model here.
In the terminology of the original register automata defined in~\cite{kaminski-francez},
it is simply register automata extended with testing whether 
the data value currently read is bigger/smaller than 
those in the registers.

It is shown in~\cite{fo2-amal} that
the satisfiability problem for $\FO^2(+1,\dvSucc)$
over {\em text} is decidable.
A {\em text} is simply a data word in which
all the data values are different and
they range over the positive integers from $1$ to $n$,
for some $n \geq 1$.
We will see later that the satisfiability problem for $\FO^2(+1,\dvSucc)$
can be reduced to the non-emptiness problem of our model.

In~\cite{schwentick-zeume}
it is shown that the satisfiability problem of the logic $\FO^2(<,\prec)$
on {\em words} is decidable.
This logic is incomparable with our model.
However, it should be noted that $\FO^2(<)$ {\em cannot} capture
the whole class of regular languages.

The work on data trees that we are aware of is in~\cite{BMSS09,lazic-lics07}.
In~\cite{BMSS09} it was shown that
the satisfiability problem for the logic $\EMSO^2(\eda,\era,\sim)$
over unranked trees is decidable in 3-$\nexp$.
However, no automata model is provided.
We will see later how this logic corresponds precisely to
a special subclass of ODTA.

In~\cite{lazic-lics07} alternating tree register automata were introduced for trees.
They are essentially the generalisation of the alternating $\RA_1$
to the tree case.
It was shown that this model captures the forward XPath queries.
However, no logical characterisation is provided
and the non-emptiness problem, though decidable, is non primitive recursive.

As mentioned earlier, the main idea in this paper is
the conversion of the data values from an infinite domain back to string over a finite alphabet.
Roughly speaking, it works as follows.
Given an ordered-data tree $t$, we show how to 
construct a string $w$ over a finite alphabet
whose domain corresponds precisely to the data values in $t$.
We then use the classical finite state automaton
to reason about $w$, and thus, also about the data values in $t$.
This idea is the main difference between our paper and the existing works.
Most of the existing techniques rely on some specific combinatorial properties
of the formalisms considered, which make them highly independent of one another.
As we will see later, our model captures quite a few other formalisms
without significant jump in complexity.

\paragraph*{Organisation}
This paper is organised as follows.
In Section~\ref{s: preliminaries} we give some preliminary background.
In Section~\ref{s: ordered-data tree}
we formally define the logic for ordered-data trees and present a few examples
as well as notations that we need in this paper.
In Section~\ref{s: tool} we present two lemmas that we are going to need later on. 
We prove them in a quite general setting,
as we think they are interesting in their own.
We introduce the ordered-data tree automata (ODTA) in Section~\ref{s: S-automata}
and weak ODTA in Section~\ref{s: weak S-automata}.
In Section~\ref{s: undecidable}
we discuss a couple of the undecidable extensions of weak ODTA.
In Section~\ref{s: string data value}
we describe how to modify the definition of ODTA
when the data values are strings, that is,
when they come from a partially ordered domain.
Finally we conclude with some concluding remarks in Section~\ref{s: conclusion}.

\section{Preliminaries}
\label{s: preliminaries}

In this section we review some definitions that
we are going to use later on.
We usually use $\Gamma$ and $\Sigma$ 
to denote finite alphabets.
We write $2^{\Gamma}$ to denote
an alphabet in which each symbol corresponds to a subset of $\Gamma$.
In some cases, we may need the alphabet $2^{2^{\Gamma}}$ --
an alphabet in which each symbol corresponds to a set of subsets of $\Gamma$.
We denote the set of natural numbers $\{0,1,2,\ldots\}$ by $\nn$.

Usually we write $\L$ to denote a language, for both string and tree languages.
When it is clear from the context, 
we use the term {\em language} to mean either a string language,
or a tree language.

\subsection{Finite state automata over strings and commutative regular languages}

We usually write $\M$ to denote a finite state automaton on strings.
The language accepted by the automaton $\M$ is denoted by $\L(\M)$.

Let $\Sigma = \{a_1,\ldots,a_{\ell}\}$.
For a word $w \in \Sigma^{\ast}$,
the Parikh image of $w$ is $\Parikh(w)=(n_1,\ldots,n_{\ell})$,
where $n_i$ is the number of appearances of $a_i$ in $w$.
For a vector $\bar{n}$,
the inverse of the Parikh image of $\bar{n}$ is 
$\Parikh^{-1}(\bar{n}) = \{w \mid w \in \Sigma^{\ast} \ \mbox{and} \ \Parikh(w)=\bar{n}\}$.

For $1\leq i \leq \ell$,
a vector $\bar{v} = (n_1,\ldots,n_{\ell}) \in \nn^{\ell}$
is called an {\em $i$-base},
if $n_i \neq 0$ and $n_j = 0$, for all $j\neq i$.
A language $\L$ is {\em periodic},
if there exist $(\ell+1)$ vectors $\bar{u},\bar{v}_1,\ldots,\bar{v}_{\ell}$
such that $\bar{u} \in \nn^{\ell}$ and each $\bar{v}_i$ is an $i$-base
and 
$$
\L = \bigcup_{h_1,\ldots,h_{\ell}\geq 0} 
\Parikh^{-1}(\bar{u} + h_1\bar{v}_1 + \cdots + h_{\ell} \bar{v}_{\ell}).
$$
We denote such language $\L$ by $\L(\bar{u},\bar{v}_1,\ldots,\bar{v}_{\ell})$.

A language $\L$ is {\em commutative} if it is closed under reordering.
That is, if $w = b_1\cdots b_m \in \L$,
and $\sigma$ is a permutation on $\{1,\ldots,m\}$,
then $b_{\sigma(1)}\cdots b_{\sigma(m)} \in \L$.

The following result is a kind of folklore and can be proved easily.
\begin{theorem}
\label{t: commutative regular}
A language is commutative and regular if and only if
it is a finite union of periodic languages.
\end{theorem}

\subsection{Unranked trees, tree automata and transducers}
\label{ss: unranked trees}

An unranked
finite tree domain is a prefix-closed finite subset $D$ of $\nn^*$
(words over $\nn$) such that $u \cdot i \in D$ implies $u\cdot j\in D$
for all $j < i$ and $u\in\nn^*$. Given a finite labeling alphabet
$\Sigma$, a $\Sigma$-labeled unranked tree $t$ is a structure 
$$
\langle D,\eda, \era, \{a(\cdot)\}_{a\in\Sigma}\rangle,
$$ where
\begin{itemize}\itemsep=0pt
\item
$D$ is an unranked tree domain,
\item
$\eda$ is the child relation:
$(u, u\cdot i) \in \eda$ for all $u, u\cdot i\in D$,
\item
$\era$ is the next-sibling relation:
$(u \cdot i, u \cdot (i+1))\in\era$ 
for all $u\cdot i, u\cdot (i+1)\in D$, and
\item
the $a(\cdot)$'s are labeling predicates, i.e. for each node $u$,
exactly one of $a(u)$, with $a\in\Sigma$, is true.
\end{itemize}
We write $\Domain(t)$ to denote the domain $D$.
The label of a node $u$ in $t$ is denoted by $\lab_t(u)$.
If $\lab_t(u) = a$, then we say that $u$ is an $a$-node.

An {\em unranked tree automaton}~\cite{tata,thatcher67} over
$\Sigma$-labeled trees is a tuple $\A=\langle Q,\Sigma,\delta,F\rangle$, where $Q$
is a finite set of states, $F\subseteq Q$ is the set of final states,
and $\delta: Q\times\Sigma\to 2^{(Q^*)}$ is a transition function; we
require $\delta(q,a)$'s to be  regular languages over $Q$ for all
$q\in Q$ and $a\in\Sigma$.

A run of $\A$ over a tree $t$ is a function $\rho_{\sA}: \Domain(t) \to
Q$ such that for each node $u$ with $n$ children 
$u\cdot 0,\ldots, u \cdot (n-1)$, 
the word $\rho_{\sA}(u\cdot 0)\cdots \rho_{\sA}(u\cdot (n-1))$ 
is in the language $\delta(\rho_{\sA}(u),\lab_t(u))$.  
For a leaf $u$ labeled $a$, this means that $u$ could be assigned a state
$q$ if and only if the empty word $\epsilon$ is in $\delta(q,a)$.  A run is
accepting if $\rho_{\sA}(\epsilon)\in F$, i.e., if the root is assigned
a final state.  A tree $t$ is accepted by $\A$ if there exists
an accepting run of $\A$ on $t$.  The set of all trees accepted by $\A$ is
denoted by $\L(\A)$.

An unranked tree (letter-to-letter) transducer with the input alphabet $\Sigma$
and output alphabet $\Gamma$
is a tuple $\T=\langle \A,\mu \rangle$, where $\A$ is a tree automaton
with the set of states $Q$,
and $\mu\subseteq Q\times\Sigma\times\Gamma$ is an output relation.
We call such $\T$ a transducer from $\Sigma$ to $\Gamma$.

Let $t$ be a $\Sigma$-labeled tree, and $t'$ a $\Gamma$-labeled tree
such that $\Domain(t) = \Domain(t')$.
We say that a tree $t'$ is an output of $\T$ on $t$, if
there is an accepting run $\rho_{\sA}$ of $\A$ on $t$ and 
for each $u \in \Domain(t)$, 
it holds that $(\rho_{\sA}(u),\lab_{t}(u),\lab_{t'}(u)) \in \mu$.
We call $\T$ an identity transducer, if
$\lab_{t}(u)=\lab_{t'}(u)$ for all $u\in \Domain(t)$.
We will often view an automaton $\A$ as an identity transducer.

\subsection{Automata with Presburger constraints (APC)}

An automaton with Presburger constraints (APC) 
is a tuple $\langle \A,\xi \rangle$, where $\A$ is an unranked tree automaton with
states $q_0,\ldots,q_m$ and $\xi$ is an existential Presburger formula with 
free variables $x_0,\ldots,x_m$.
A tree $t$ is accepted by $\langle\A,\xi\rangle$, denoted by $t \in \L(\A,\xi)$, 
if there is an accepting run $\rho_{\sA}$ of $\A$ on $w$ such that 
$\xi(n_0,\ldots,n_m)$ is true, where $n_i$ is the number of
appearances of $q_i$ in $\rho_{\sA}$.

\begin{theorem}
\label{t: non-emptiness and closure presburger}
{\em \cite{schwentick-icalp04,Schwentick-cade-05}}
The non-emptiness problem for APC is decidable in $\np$.
\end{theorem}

It is worth noting also that
the class of languages accepted by APC is closed under union and intersection. 

Oftentimes, instead of counting the number of states
in the accepting run,
we need to count the number of occurrences of alphabet symbols in the tree.
Since we can easily embed the alphabet symbols inside the states,
we always assume that the Presburger formula $\xi$
has the free variables $x_a$'s to denote the 
number of appearances of the symbol $a$ in the tree.

As in the word case, we let $\Parikh(t)$ denote the Parikh image of the tree $t$.
We will need the following proposition.

\begin{proposition}~{\em \cite{schwentick-icalp04,Schwentick-cade-05}}
\label{p: aut-to-presburger}
Given an unranked tree automaton $\A$,
one can construct, in polynomial time, an existential Presburger formula 
$\xi_{\sA}(x_1,\ldots,x_{\ell})$ such that
\begin{itemize}\itemsep=0pt
\item 
for every tree $t\in \L(\A)$, $\xi_{\sA}(\Parikh(t))$ holds;
\item
for every $\bar{n}=(n_1,\ldots,n_{\ell})$ such that $\xi_{\sA}(\bar{n})$ holds,
there exists a tree $t \in \L(\A)$ with $\Parikh(t)=\bar{n}$.
\end{itemize}
\end{proposition}

%In fact, Theorem~\ref{t: non-emptiness and closure presburger}
%can be proved via Proposition~\ref{p: aut-to-presburger} as follows.
%Given an input APC $(\A,\xi)$, convert $\A$ into its existential Presburger formula $\xi_{\sA}$
%via Proposition~\ref{p: aut-to-presburger}, 
%and then test the satisfiability of $\xi_{\sA}\wedge \xi$,
%which can be done in $\np$.

\section{Ordered-data trees and Their Logic}
\label{s: ordered-data tree}

An ordered-data tree over the alphabet $\Sigma$ is a tree
in which each node, besides carrying a label from the finite alphabet
$\Sigma$, also carries a data value from $\bbN = \{0,1,\ldots\}$.\footnote{Here
we use the natural numbers as data values just to be concrete.
The results in our paper apply trivially for any linearly ordered domain.}

Let $t$ be an ordered-data tree over $\Sigma$
and $u \in \Domain(t)$.
We write $\val_t(u)$ to denote the data value in the node $u$. 
The set of all data values in the $a$-nodes in $t$ is denoted by $V_t(a)$.
That is, $V_t(a)=\{\val_t(u)\ | \ \lab_t(u)=a \ \mbox{and} \ u \in \Domain(t)\}$.
We write $V_t$ to denote the set of data values
found in the tree $t$.
We also write $\#_t(a)$ to denote the number of $a$-nodes in $t$.

The profile of a node $u$ is a triplet 
$(l,p,r) \in \{\top,\bot,*\}\times\{\top,\bot,*\}\times\{\top,\bot,*\}$,
where $l=\top$ and $l=\bot$ indicate that 
the node $u$ has the same data value and different data value as its left sibling, respectively;
$l=*$ indicates that $u$ does not have a left sibling.
Similarly, $p=\top$, $p=\bot$, and $p=*$
have the same meaning in relation to the parent of the node $u$,
while $r=\top$, $r=\bot$, and $r=*$ means the same in relation to the right sibling of the node $u$.
For an ordered-data tree $t$ over $\Sigma$,
the profile tree of $t$, denoted by $\sfProfile(t)$,
is a tree over $\Sigma\times\{\top,\bot,*\}^3$
obtained by augmenting to each node of $t$ its profile. 

We write $\sfProj(t)$ to denote the $\Sigma$ projection of 
the ordered-data tree $t$,
that is, $\sfProj(t)$ is $t$ without the data values.
When we say that an ordered-data tree $t$ is accepted by
an automaton $\A$, we mean that $\sfProj(t)$ is accepted by $\A$.
An ordered-data tree $t'$ is an output of a transducer $\T$ on
an ordered-data tree $t$, if $\sfProj(t')$ is an output of $\T$ on $\sfProj(t)$,
and for all $u \in \Domain(t')$, 
we have $\val_{t'}(u) = \val_{t}(u)$.

Figure~\ref{fig: profile} shows an example of an
ordered-data tree $t$ over the alphabet $\{a,b,c\}$
with its profile tree.
The notation ${a \choose d}$ means that
the node is labeled with $a$ and has data value $d$.

\begin{figure*}
{\footnotesize
\begin{picture}(300,175)(-140,-120)

\put(-50,40){${a \choose 2}$}
\put(-44,32){\vector(-2,-1){60}}
\put(-44,32){\vector(-3,-4){24}}
\put(-44,32){\vector(3,-4){24}}
\put(-44,32){\vector(2,-1){60}}

\put(-113,-10){${b \choose 1}$}
\put(-75,-10){${c \choose 2}$}
\put(-27,-10){${a \choose 4}$}
\put(10,-10){${a \choose 6}$}

\put(-69,-18){\vector(-3,-4){24}}
\put(-69,-18){\vector(3,-4){24}}
\put(-69,-18){\vector(2,-1){64}}

\put(-99,-62){${b \choose 2}$}
\put(-51,-62){${b \choose 4}$}
\put(-10,-62){${a \choose 7}$}

\put(-93,-70){\vector(0,-1){30}}
\put(-45,-70){\vector(0,-1){30}}
\put(-4,-70){\vector(0,-1){30}}

\put(-100,-110){${c \choose 1}$}
\put(-52,-110){${c \choose 6}$}
\put(-10,-110){${b \choose 7}$}

%%%%%%%%%%%%%%%%%%%%%%%%%%%%%%%%%%%%%%%%%%%%%%%%%%%%%%%%%%%%%%%%%%%%%%%%
%%%%%%%%%%%%%%%%%%%%%%%%%%%% PROFILE TREE %%%%%%%%%%%%%%%%%%%%%%%%%%%%%%
%%%%%%%%%%%%%%%%%%%%%%%%%%%%%%%%%%%%%%%%%%%%%%%%%%%%%%%%%%%%%%%%%%%%%%%%

\put(127,38){\small ${a,(*,*,*) \choose 2}$}
\put(146,32){\vector(-2,-1){60}}
\put(146,32){\vector(-3,-4){24}}
\put(146,32){\vector(3,-4){24}}
\put(146,32){\vector(2,-1){60}}

\put(55,-12){\small ${b,(*,\bot,\bot) \choose 1}$}
\put(102,-12){\small ${c,(\bot,\top,\bot) \choose 2}$}
\put(150,-12){\small ${a,(\bot,\bot,\bot) \choose 4}$}
\put(198,-12){\small ${a,(\bot,\bot,*) \choose 6}$}

\put(121,-18){\vector(-3,-4){24}}
\put(121,-18){\vector(3,-4){24}}
\put(121,-18){\vector(2,-1){64}}

\put(76,-64){\small ${b,(*,\top,\bot) \choose 2}$}
\put(128,-64){\small ${b,(\bot,\bot,\bot) \choose 4}$}
\put(180,-62){\small ${a,(\bot,\bot,*) \choose 7}$}

\put(97,-70){\vector(0,-1){30}}
\put(145,-70){\vector(0,-1){30}}
\put(197,-70){\vector(0,-1){30}}

\put(76,-110){\small ${c,(*,\bot,*) \choose 1}$}
\put(128,-110){\small ${c,(*,\bot,*) \choose 6}$}
\put(180,-110){\small ${b,(*,\top,*) \choose 7}$}

\end{picture}
}
\caption{An example of an ordered-data tree (on the left) and its profile (on the right).}
\label{fig: profile}
\end{figure*}
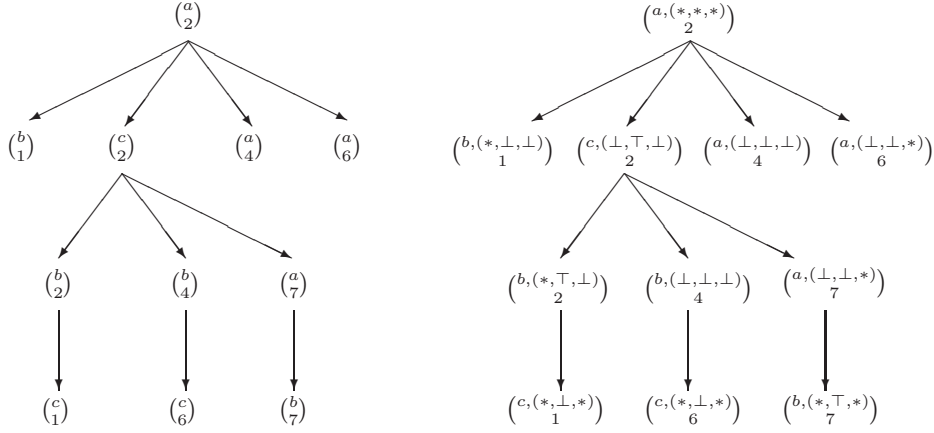

\subsection{String representations of data values}
\label{ss: string representation}

Let $t$ be an ordered-data tree over $\Gamma$.
For a set $S \subseteq \Gamma$, let
$$
[S]_t = \bigcap_{a \in S} V_t(a) \cap \bigcap_{b \notin S} \overline{V_t(b)}.
$$
That is, $[S]_t$ is the set of data values that are found in
$a$-positions for all $a\in S$ but are not found in any $b$-position for
$b\not\in S$. Note that the sets $[S]_t$'s are disjoint, and that
for each $a \in \Gamma$, 
$$
V_t(a) = \bigcup_{S \ \mbox{\scriptsize s.t.} \ a \in S} \ [S]_t.
$$
Moreover,
$|V_t(a)| = \sum_{S \ \mbox{\scriptsize s.t.} \ a \in S} \ |[S]_t|$.

Let $d_1 < \cdots < d_m$ be all the data values found in $t$.
The {\em string representation} of the data values in $t$, denoted by $\V_{\Gamma}(t)$, 
is the string $S_1\cdots S_m$ over the alphabet $2^{\Gamma} - \{\emptyset\}$ of length $m$
such that $d_i \in [S_i]_t$, for each $i=1,\ldots,m$.
The notation $[S]_t$ is already introduced in~\cite{fo2-lpar,edt},
but not $\V_{\Gamma}(t)$.

Consider the example of the tree $t$ in Figure~\ref{fig: profile}.
The data values in $t$ are $1,2,4,6,7$, where 
\begin{eqnarray*}
~[\{b,c\}]_t  & = &  \{1\},
\\ 
~[\{a,b,c\}]_t  & = &  \{2\},
\\
~[\{a,b\}]_t & = &  \{4,7\},
\\
~[\{a,c\}]_t  & = &  \{6\},
\\
~[S]_t & = & \emptyset, \ \mbox{for all the other} \ S\mbox{'s}.
\end{eqnarray*}
The string $\V_{\Gamma}(t)$ is $S_1 \ S_2 \ S_3 \ S_4 \ S_5$,
where $S_1 = \{b,c\}$, $S_2 = \{a,b,c\}$, $S_3 = S_5 = \{a,b\}$ and $S_4 = \{a,c\}$.

\subsection{A logic for ordered-data trees}
\label{ss: logic}

An ordered-data tree $t$ over the alphabet $\Sigma$ 
can be viewed as a structure
$$
t\ = \langle D, \{a(\cdot)\}_{a\in\Sigma}, \eda,\era, \sim, \prec, \dvSucc \rangle,
$$
where 
\begin{itemize}
\item
the relations $\{a(\cdot)\}_{a\in\Sigma}, \eda,\era$ 
are as defined before in Subsection~\ref{ss: unranked trees}, 
\item
$u \sim v$ holds, if $\val_t(u) = \val_t(v)$,
\item
$u \prec v$ holds, if $\val_t(u) < \val_t(v)$,
\item
$u \dvSucc\ v$ holds, if $\val_t(v)$ is the minimal data value in $t$
greater than $\val_t(u)$.
\end{itemize}
Obviously, $x \dvSucc\ y$ can be expressed equivalently 
as $x \prec y \wedge \forall z (\neg (x \prec z \wedge  z \prec y))$.
We include $\dvSucc$ for the sake of convenience.
We also assume that we have the predicates
$\root(x)$, $\fs(x)$, $\ls(x)$, and $\leaf(x)$
which stand for $\forall y (\neg \eda(y,x))$, 
$\forall y (\neg\era(y,x))$,
$\forall y (\neg\era(x,y))$, and
$\forall y (\neg\eda(x,y))$, respectively.
We also write $x\nsim y$ to denote $\neg (x\sim y)$.

For $\O \subseteq \{\eda,\era,\sim,\prec,\dvSucc\}$,
we let $\FO(\O)$ stand for the first-order logic with the vocabulary $\O$, 
$\MSO(\O)$ for its monadic second-order logic 
(which extends $\FO(\O)$ with quantification over sets of nodes),
and $\EMSO(\O)$ for its existential monadic second order logic, i.e., 
formulas of the form $\exists X_1 \ldots \exists X_m\ \psi$, 
where $\psi$ is an $\FO(\O)$ formula over the vocabulary $\O$
extended with the unary predicates $X_1,\ldots,X_m$.

We let $\FO^2(\O)$ stand for $\FO(\O)$ with two variables,
i.e., the set of $\FO(\O)$ formulae that only use two variables $x$ and $y$.
The set of all formulae of the form $\exists X_1 \ldots \exists X_m\ \psi$, 
where $\psi$ is an $\FO^2(\O)$ formula is denoted by $\EMSO^2(\O)$. 
Note that $\EMSO^2(\eda,\era)$ is equivalent in expressive power 
to $\MSO(\eda,\era)$ over the usual (without data) trees.
That is, it defines precisely the regular tree languages~\cite{thomas-handbook}.

As usual, we define $\L_{data}(\varphi)$ as the
set of ordered-data trees that satisfy the formula $\varphi$.
In such case, we say that the formula $\varphi$ 
expresses the language $\L_{data}(\varphi)$.\footnote{To avoid confusion,
we put the subscript $data$ on $\L_{data}$ to denote a 
language of ordered-data trees. We use the symbol $\L$ without
the subscript $data$ to denote the usual language of trees/strings without data.}

The following theorem is well known. 
It shows how even extending $\FO(\eda,\era)$ with 
equality test on data values
immediately yields undecidability. 

\begin{theorem}
\label{t: fo undecidable}
{\em (See, for example,~\cite{NSV04})}
The satisfiability problem for the logic $\FO(\eda,\era,\sim)$ is undecidable.
\end{theorem}

One of the deepest results in this area is the following decidability
result for the logic $\EMSO^2(\eda,\era,\sim)$.
\begin{theorem}
\label{t: emso2 decidable}
{\em \cite{BMSS09}}
The satisfiability problem for the logic $\EMSO^2(\eda,\era,\sim)$ is decidable.
\end{theorem}

\subsection{A few examples}
\label{ss: logic example}

In this subsection we present a few examples of 
properties of ordered-data trees.
Some of them are special cases
of more general techniques that will be used later on.

\begin{example}
\label{eg: two a-nodes}
Let $\Sigma = \{a,b\}$.
Consider the language $\L_{data}^a$ of ordered-data trees over $\Sigma$ where
an ordered-data tree $t\in \L_{data}^{a}$ if and only if
there exist two $a$-nodes $u$ and $v$ such that
$u$ is an ancestor of $v$ and either $v \sim u$ or $v \prec u$.
This language can be expressed with the formula $\exists X \exists Y \exists Z \ \varphi$,
where $\varphi$ states that
$X$ contains only the node $u$,
$Y$ contains only the node $v$,
$Z$ contains precisely the nodes in the path from $u$ to $v$,
and $v \sim u$ or $v \prec u$.
\end{example}

\begin{example}
\label{eg: m data values}
For a fixed set $S\subseteq \Sigma$ and an integer $m \geq 1$,
we consider the language $\L_{data}^{S,m}$ such that
$t\in \L_{data}^{S,m}$ if and only if $|[S]_t| = m$.

We pick an arbitrary symbol $a\in S$.
The language $\L_{data}^{S,m}$ can be expressed in $\EMSO^2(\eda,\era,\sim)$
with the formula of the form
$\exists X_1 \ \cdots \exists X_m \ \varphi$,
where $\varphi$ is a conjunction of the following.
\begin{itemize}
\item
That the predicates $X_1,\ldots,X_m$ are disjoint and
each of them contains exactly one node, which is an $a$-node. 
\item
That the data values found in nodes in $X_1,\ldots,X_m$ are all different.
\item
That for each $i\in\{1,\ldots,m\}$,
if a data value is found in a node in $X_i$, 
then it must also be found in some $b$-node, for every $b\in S$.
\item
That for each $i\in\{1,\ldots,m\}$,
if a data value found in a node in $X_i$, 
then it must {\em not} be found in any $b$-node, for every $b\notin S$.
\item
That for every $a$-node (recall that $a\in S$)
that does not belong to the $X_i$'s, either
it has the same data value as the data value in a node belongs to one of the $X_i$'s, or
it has the data value {\em not} in $[S]_t$.
\\
That its data value does not belong to $[S]_t$ can be stated as the negation of  
\begin{itemize}
\item
for each $b\in S$,
there is a $b$-node with the same data value; and
\item
the data value cannot be found in any $b$-node, for every $b\notin S$. 
\end{itemize}
\end{itemize}
To express all these intended meanings,
it is sufficient that $\varphi \in \FO^2(\eda,\era,\sim)$.
\end{example}

\begin{example}
\label{eg: mod m data values}
For a fixed set $S\subseteq \Sigma$ and an integer $m \geq 1$,
we consider the language $\L_{data}^{S,\pmod{m}}$ such that
$t\in  \L_{data}^{S,\pmod{m}}$ if and only if $|[S]_t| \equiv 0 \pmod{m}$.

This language $\L_{data}^{S,\pmod{m}}$ can be expressed in $\EMSO^2(\eda,\era,\sim)$
with a formula of the form
$$
\exists X_0 \cdots \exists X_{m-1} 
\exists Y_0 \cdots \exists Y_{m-1}  \exists Z \
\psi,
$$
where the intended meanings of $X_0,\ldots,X_{m-1},Y_0,\ldots,Y_{m-1},Z$ are as follows.
For a node $u$ in an ordered-data tree $t \in \L_{data}$, 
\begin{itemize}
\item
the number of nodes belonging to $Z$ is precisely $|[S]_t|$; and
if $Z(u)$ holds in $t$, then
the data value in the node $u$ belongs to $[S]_t$;
\item 
$X_i(u)$ holds in $t$ if and only if
in the subtree $t'$ rooted in $u$
we have $|[S]_{t'}| \equiv i \pmod{m}$;
\item
if $v_1,\ldots,v_k$ are all the left-siblings of $u$,
and $X_{i_1}(v_1),\ldots,X_{i_k}(v_k)$ holds,
then 
$Y_i(u)$ holds if and only if
$i_1 + \cdots + i_k \equiv i \pmod{m}$.
\end{itemize}
To express all these intended meanings,
it is sufficient that $\psi \in \FO^2(\eda,\era,\sim)$.
\end{example}

\begin{example}
\label{eg: max a-node}
Let $\Sigma = \{a,b\}$.
Consider the language $\L_{data}^{a\ast}$ of ordered-data trees over $\Sigma$ where
an ordered-data tree $t\in \L_{data}^{a\ast}$ if and only if
all the $a$-nodes with data values different from the ones in their parents
satisfy the following conditions:
\begin{itemize}
\item
the data values found in these nodes are all different;
\item
one of the these data values must be the largest in the tree $t$.
\end{itemize}
The language $\L_{data}^{a\ast}$ can be expressed in $\EMSO^2(\eda,\sim,\prec)$ 
with the following formula:
\begin{eqnarray*}
\exists X  & \Big( & \forall x \big( X(x) \iff a(x) \wedge \exists y (\eda(y,x) \wedge y \nsim x)\big)
\\
& & \wedge \ \ \forall x \forall y (X(x) \wedge X(y) \wedge x\sim y \to x=y)
\\
& & \wedge \ \ \exists x \big(X(x) \wedge \forall y (y \prec x \vee x \sim y)\big) \Big).
\end{eqnarray*}
\end{example}

\section{Two Useful Lemmas}
\label{s: tool}

In this section we prove two lemmas which will be used later on.
The first is combinatorial by nature, and we will use it
in our proof of the decidability of ODTA.
The second is an Ehrenfeucht-Fra\"iss\'e type lemma for ordered-data trees,
and we will use it in our proof of the logical characterization of ODTA.

\subsection{A combinatorial lemma}

Let $G$ be an (undirected and finite) graph.
For simplicity, we consider only the graph without self-loop.
We denote by $V(G)$ the set of vertices in $G$
and $E(G)$ the set of edges.
For a node $u \in V(G)$,
we write $\deg(u)$ to denote the degree of the node $u$
and $\deg(G)$ to denote $\max\{\deg(u) \mid u \in V(G)\}$.

A {\em data graph} over the alphabet $\Gamma$ is
a graph $G$ in which each node carries a label from $\Gamma$
and a data value from $\nn$.
A node $u\in V(G)$ is called
an $a$-node, if its label is $a$, in which case we write $\lab_G(u)=a$.
We denote by $\val_{G}(u)$ the data value found in node $u$,
and $\Val_G(a)$ the set of data values found in $a$-nodes in $G$.

\begin{lemma}
\label{l: canonical lemma}
Let $G$ be a data graph over $\Gamma$.
Suppose for each $a \in \Gamma$, we have $|V_G(a)| \geq \deg(G)|\Gamma|+\deg(G)+1$.
Then we can reassign the data values in the nodes in $G$ to obtain
another data graph $G'$ such that $V(G)=V(G')$ and $E(G)=E(G')$ and
\begin{itemize}
\item[(1)]
for each $u \in V(G')$, $\lab_G(u) = \lab_{G'}(u)$;
\item[(2)]
for each $a \in \Gamma$, $\Val_{G}(a) = \Val_{G'}(a)$;
\item[(3)]
for each $u,v \in V(G)$,
if $(u,v) \in E(G')$,
then $\val_{G'}(u)\neq \val_{G'}(v)$.
\end{itemize} 
\end{lemma}
\begin{proof}
Note that in the lemma
the data graph $G'$ differs from $G$ only in the data values on the nodes,
where we require that adjacent nodes in $G'$ have different data values.

In the following we write $\#_G(a)$ to denote the number of $a$-nodes in $G$
and $K=\deg(G)$.
First, we perform some partial reassignment of the data values on some nodes.
For each $a \in \Gamma$, we pick $|\Val_G(a)|$ number of $a$-nodes in $G'$.
Then we assign to these $a$-nodes the data values from $\Val_G(a)$.
One $a$-node gets one data value.
Such assignment can be done since obviously $\#_{G}(a) \geq |\Val_G(a)|$.
If $\#_G(a)> |\Val_G(a)|$, then there will be some $a$-nodes in $G'$
that do not have data values.
We write $\val_{G'}(u) = \sharp$, if $u$ does not have data value.
From this step we already obtain that
$\Val_{G'}(a) = \Val_{G}(a)$ for each $a\in\Gamma$.

However, reassigning the data values just like that,
there may exist an edge $(u,v)$ such that
$\val_{G'}(u) = \val_{G'}(v)$ and $\val_{G'}(u),\val_{G'}(v)\neq \sharp$.
We call such an edge a {\em conflict} edge.
We are going to reassign the data values one more time so that there is no conflict edge.

Suppose there exists an edge $(u,v)\in E$ such that $\val_{G'}(u) = \val_{G'}(v)= d$ and 
suppose that $u$ is an $a$-node, for some $a\in\Gamma$.
The data value $d$ can only be found in at most $|\Gamma|$ nodes in $G'$.
Since $\deg(G)=K$, the neighbours of those nodes (with data value $d$) are at most $K|\Gamma|$ nodes.
Now $|\Val_G(a)| = |\Val_{G'}(a)| \geq K|\Gamma|+K+1$,
there are at least $K+1$ number of $a$-nodes
whose neighbours do not get the data value $d$.
Let $u_1,\ldots,u_m$ be such $a$-nodes, where $m \geq K+1$.
From these nodes, there exists $i$ such that
$\val_{G'}(u_i)\notin \{\val_{G'}(x) \mid (u,x)\in E\}$.

We can then swap the data values on the nodes $u$ and $u_i$,
and this results in one less conflict edge.
We repeat this process until there is no conflict edge.
Now it is straightforward that 
\begin{itemize}
\item[(1)]
for each $u \in V$, $\lab_G(u) = \lab_{G'}(u)$;
\item[(2)]
for each $a \in \Gamma$, $\Val_{G}(a) = \Val_{G'}(a)$;
\item[(3)]
for each $u,v \in V$,
if $(u,v) \in E$ and $\val_{G'}(u),\val_{G'}(v)\neq \sharp$, 
then $\val_{G'}(u)\neq \val_{G'}(v)$.
\end{itemize} 
What is left to do now is to assign data values to the nodes $u$,
where $\val_{G'}(u)=\sharp$.
For each $a$-node, where $\val_{G'}(u)=\sharp$,
we pick the data value $d\in \Val_{G'}(a) = \Val_{G}(a)$
which is not assigned to any its neighbour.
Such data value exists since $|\Val_{G'}(a)|\geq K |\Gamma|+K+1 \geq K+1$.
Such assignment will not violate condition~(3) above,
thus, we get the desired data graph $G'$.
This completes the proof of Lemma~\ref{l: canonical lemma}.
\end{proof}

\subsection{An Ehrenfeucht-Fra\"iss\'e type lemma}

We need the following notation.
A $k$-characteristic function on the alphabet $\Gamma$,
is a function $f: \Gamma \to \{0,1,2,\ldots,k\}$.
Let $\F_{\Gamma,k}$ be the set of all such $k$-characteristic functions on $\Gamma$.
A function $f \in \F_{\Gamma,k}$ is a $k$-characteristic function for a set $S \subseteq \Gamma$,
if $f(a)\in \{1,2,\ldots,k\}$, for all $a\in S$, 
and $f(a)=0$, for all $a \notin S$.

An {\em ordered-data set} $\frU$ over the alphabet $\Gamma$ 
consists of a finite set $U$, in which
each element $u \in U$ carries a label $\lab_{\sfrU}(u) \in \Gamma$
and a data value $\val_{\sfrU}(u) \in \bbN$.
An element $u \in U$ is called an $a$-element, if $\lab_{\sfrU}(u)=a$.
In other words, an ordered-data set is similar to an ordered-data tree, 
but without the relations $\eda$ and $\era$.
It can be viewed as a structure
$\frU  =  \langle U, \{a(\cdot)\}_{a\in\Gamma}, \sim, \prec, \dvSucc \rangle$,
where 
\begin{itemize}\itemsep=0pt
\item
for each $a\in \Gamma$ and $u \in U$,
the relation $a(u)$ holds if $\lab_{\sfrU}(u)=a$,
\item
$u \sim v$ holds, if $\val_{\sfrU}(u) = \val_{\sfrU}(v)$,
\item
$u \prec v$ holds, if $\val_{\sfrU}(u) < \val_{\sfrU}(v)$,
\item
$u \dvSucc\ v$ holds, if $\val_{\sfrU}(v)$ is the minimal data value found in $\frU$
greater than $\val_{\sfrU}(u)$.
\end{itemize}
Let $\frU$ be an ordered-data set and
$d_1 < \cdots < d_m$ be the data values found in $\frU$.
The $k$-{\em extended representation} of $\frU$ is the string 
$\V^k_{\Gamma}(\frU) = (S_1,f_1)\cdots (S_m,f_m) \in 2^{\Gamma}\times \F_{\Gamma,k}$
such that 
$S_1 \cdots S_m = \V_{\Gamma}(\frU)$ and 
for each $i\in \{1,2,\ldots,m\}$ and for each $a \in \Gamma$,
\begin{enumerate}\itemsep=0pt
\item 
$f_i$ is a $k$-characteristic function for the set $S_i$,
\item 
if $1\leq f_i(a) \leq k-1$, then there are $f_i(a)$
number of $a$-elements in $U$ with data value $d_i$,
\item
if $f_i(a)=k$, then there are at least $k$ number of $a$-elements in $U$ 
with data value $d_i$.
\end{enumerate}

We assume that in every formula in $\MSO(\sim,\prec,\dvSucc)$
all the monadic second-order quantifiers precede the first-order part.
That is, sentences in $\MSO(\sim,\prec,\dvSucc)$ are of the form:
$\varphi := Q_1 X_1 \cdots Q_{s} X_{s} \; \psi$,
where the $X_i$'s are monadic second-order variables,
the $Q_i$'s are  $\exists$ or $\forall$
and $\psi\in\FO(\sim,\prec,\dvSucc)$
extended with the unary predicates $X_1,\ldots,X_s$.
We call the integer $s$, the MSO quantifier rank of $\varphi$,
denoted by $\sfMSOqr(\varphi)=s$,
while we write $\sfFOqr(\varphi)$ to denote
the quantifier rank of $\psi$,
that is the quantifier rank of the first-order part of $\varphi$.

\begin{lemma}
\label{l: extended rep for data set}
Let $\frU_1$ and $\frU_2$ be ordered-data sets over $\Gamma$ such that
$\V^{k 2^{s}}_{\Gamma}(\frU_1) = \V^{k2^s}_{\Gamma}(\frU_2)$.
For any $\MSO(\sim,\prec,\dvSucc)$ sentence $\varphi$
such that $\sfMSOqr(\varphi)\leq s$ and $\sfFOqr(\varphi)\leq k$,
$\frU_1 \models \varphi \quad \mbox{if and only if} \quad \frU_2 \models \varphi$.
\end{lemma}
\begin{proof}
The proof is by Ehrenfeucht-Fra\"iss\'e game for MSO of $(s+k)$ rounds,
with $s$ rounds of set-moves and $k$ rounds of point-moves.
We can assume that the set-moves precede the point-moves.
See, for example,~\cite{libkin04}, for the definition of Ehrenfeucht-Fra\"iss\'e game.

Before we go to the proof, we need a few notations.
Let $\frU_1$ and $\frU_2$ be ordered-data sets over $\Gamma$.
For $(a,d) \in \Gamma\times\bbN$, we write 
$P_{\sfrU_1}(a,d) = \{u \mid \lab_{\sfrU_1}(u)=a \ \mbox{and} \ \val_{\sfrU_1}(u)=d\}$ --
the set of elements in $U_1$ with label $a$ and data value $d$.
We can define similarly $P_{\sfrU_2}(a,d)$ for $\frU_2$.

Let $\O \subseteq \{\sim,\prec,\dvSucc\}$.
Let $u_1,\ldots,u_k \in U_1$
and $v_1,\ldots,v_k \in U_2$, for some ordered-data sets $U_1$ and $U_2$.
The mapping $(u_1,\ldots,u_k)\mapsto (v_1,\ldots,v_k)$
is a partial $\O$-isomorphism (with equality) from $U_1$ to $U_2$,
if it is a partial isomorphism with regards to the vocabulary $\O$,
and if $u_l=u_{l'}$, then $v_l=v_{l'}$.

We are going to describe Duplicator's strategy for winning the 
Ehrenfeucht-Fra\"iss\'e game for MSO of $s$ rounds of set-moves,
followed by $k$ rounds of point moves.
We start with the set-moves.

\vspace{0.7 em}
\par\noindent
\underline{\em Duplicator's strategy for set-moves}:
Suppose that the game is already played for $l$ rounds,
where $X_1,\ldots,X_l$ and $Y_1,\ldots,Y_l$
are the sets of positions chosen in $U_1$ and $U_2$, respectively.
For each $I \subseteq \{1,\ldots,l\}$,
define the following set:
\begin{eqnarray*}
P_{\sfrU_1}(a,d;I) & = &
P_{\sfrU_1}(a,d) \cap \bigcap_{i \in I} X_i \cap \bigcap_{j \notin I} \overline{X_j}
\\
P_{\sfrU_2}(a,d;I) & = &
P_{\sfrU_2}(a,d) \cap \bigcap_{i \in I} Y_i \cap \bigcap_{j \notin I} \overline{Y_j}
\end{eqnarray*}
Duplicator's strategy is to preserve the following identity:
for every $(a,d) \in \Gamma \times \bbN$ and every $I \subseteq \{1,\ldots,l\}$
\begin{itemize}
\item 
If the cardinality $|P_{\sfrU_1}(a,d;I)| < k 2^{m-l}$,
then $|P_{\sfrU_1}(a,d;I)| = |P_{\sfrU_2}(a,d;I)|$.
\item 
If the cardinality $|P_{\sfrU_1}(a,d;I)| \geq k 2^{m-l}$,
then also $|P_{\sfrU_2}(a,d;I)| \geq k 2^{m-l}$.
\end{itemize}
Now suppose that on the $(l+1)^{\rm th}$ set-move,
Spoiler chooses a set $X$ of positions on $U_1$.
Duplicator chooses a set $Y$ of positions on $U_2$ as follows.
For each $I \subseteq \{1,\ldots,l\}$,
there are four cases:
\begin{enumerate}
\item 
$|P_{\sfrU_1}(a,d;I) \cap X|$ and $|P_{\sfrU_1}(a,d;I) \cap \overline{X}| < k 2^{m-l-1}$.
\\
Then, $|P_{\sfrU_1}(a,d;I)| < k 2^{m-l}$,
which by induction hypothesis, implies $|P_{\sfrU_2}(a,d;I)| = |P_{\sfrU_1}(a,d;I)|$.
Duplicator picks $|P_{\sfrU_1}(a,d;I) \cap X|$ number of points from $P_{\sfrU_2}(a,d;I)$,
and declares them ``belong to $Y$.''
The rest of the points from $P_{\sfrU_2}(a,d;I)$ are declared ``not belong to $Y$.''
\\
Obviously, $|P_{\sfrU_1}(a,d;I) \cap X| = |P_{\sfrU_2}(a,d;I) \cap Y| < k 2^{m-l-1}$
and $|P_{\sfrU_1}(a,d;I) \cap \overline{X}| = |P_{\sfrU_2}(a,d;I) \cap \overline{Y}| < k 2^{m-l-1}$.
\item 
$|P_{\sfrU_1}(a,d;I) \cap X| < k 2^{m-l-1}$ and $|P_{\sfrU_1}(a,d;I) \cap \overline{X}| \geq k 2^{m-l-1}$.
\\
In this case, either $P_{\sfrU_1}(a,d;I) < k2^{m}$ or $\geq k2^{m}$.
In either case there are $|P_{\sfrU_1}(a,d;I) \cap X|$ number of points from $P_{\sfrU_2}(a,d;I)$
which Duplicator declares as ``belong to $Y$.''
The rest of the points from $P_{\sfrU_2}(a,d;I)$ are declared ``not belong to $Y$.''
\\
Obviously, $|P_{\sfrU_1}(a,d;I) \cap X| = |P_{\sfrU_2}(a,d;I) \cap Y|$
and $|P_{\sfrU_2}(a,d;I) \cap \overline{Y}| \geq k 2^{m-l-1}$.
\item 
$|P_{\sfrU_1}(a,d;I) \cap X| \geq k 2^{m-l-1}$ and $|P_{\sfrU_1}(a,d;I) \cap \overline{X}| < k 2^{m-l-1}$.
\\
Similar to Case~2.
\item 
$|P_{\sfrU_1}(a,d;I) \cap X| \geq k 2^{m-l-1}$ and $|P_{\sfrU_1}(a,d;I) \cap \overline{X}| \geq k 2^{m-l-1}$.
\\
Then, $|P_{\sfrU_1}(a,d;I)| \geq k 2^{m-l}$, and so $|P_{\sfrU_2}(a,d;I)| \geq k 2^{m-l}$.
Duplicator declares half of the points in $P_{\sfrU_2}(a,d;I)$ as ``belong to $Y$''
and the other half as ``not belong to $Y$.''
\\
Obviously, $|P_{\sfrU_2}(a,d;I) \cap Y|$
and $|P_{\sfrU_2}(a,d;I) \cap \overline{Y}| \geq k 2^{m-l-1}$.
\end{enumerate}
Now after $m$ rounds of set-moves,
we have the following identity:
for every $(a,d) \in \Sigma \times \bbN$ and every $I \subseteq \{1,\ldots,m\}$
\begin{itemize}
\item 
If the cardinality $|P_{\sfrU_1}(a,d;I)| < k $,
then 
$|P_{\sfrU_1}(a,d;I)| = |P_{\sfrU_2}(a,d;I)|$.
\item 
If the cardinality $|P_{\sfrU_1}(a,d;I)| \geq k$,
then also $|P_{\sfrU_2}(a,d;I)| \geq k$.
\end{itemize}
This ends our description of Duplicator's strategy for set-moves.
Now we describe Duplicator's strategy for point-moves.

\vspace{0.7 em}
\par\noindent
\underline{\em Duplicator's strategy for point-moves}:
Suppose that the game is now on $l$th step.
Let $(u_1,\ldots,u_l)\mapsto (v_1,\ldots,v_l)$
be a partial $\{\sim,\prec,\dvSucc\}$-isomorphism, where $0\leq l \leq k-1$.
Suppose Spoiler chooses an element $u_{l+1}$ from $U_1$
such that $\val_{\sfrU_1}(u_{l+1})$ is the $j^{th}$ largest data value in $\frU_1$.
\begin{itemize}
\item
If $u_{l+1}=u_{l'}$, for some $l' \in \{1,\ldots,l\}$,
Duplicator chooses $v_{l+1}=v_{l'}$ from $U_2$.
\item
If $u_{l+1} \notin \{u_1,\ldots,u_l\}$, 
Duplicator chooses $v_{l+1}$ from $U_2$ such that $v_{l+1} \notin \{v_1,\ldots,v_l\}$
and $\lab_{\sfrU_1}(u_{l+1}) = \lab_{\sfrU_2}(v_{l+1})$ and $\val_{\sfrU_2}(v_{l+1})$
is the $j^{th}$ largest data value in $\frU_2$.
Such an element exists, as $\V^{k2^m}(\frU_1)=\V^{k2^m}(\frU_2)$.
\end{itemize}
In either case $(u_1,\ldots,u_{l+1})\mapsto (v_1,\ldots,v_{l+1})$
is a partial $\{\sim,\prec,\dvSucc\}$-isomorphism.
This completes the description of Duplicator's strategy and hence, our proof.
\end{proof}

Now, we define the $k$-extended representation of an ordered-data tree $t$
over the alphabet $\Gamma$,
denoted by $\V^k_{\Gamma}(t)$
is the $k$-extended representation of the ordered-data set $\frU$
obtained by ignoring the relations $\eda$ and $\era$ in $t$.
The following corollary is an immediate consequence of Lemma~\ref{l: extended rep for data set}
above.

\begin{corollary}
\label{cor: extended rep for mso}
Let $t_1$ and $t_2$ be ordered-data trees over $\Gamma$ such that
$\V^{k 2^{s}}_{\Gamma}(t_1) = \V^{k2^s}_{\Gamma}(t_2)$.
For any $\MSO(\sim,\prec,\dvSucc)$ sentence $\varphi$
such that $\sfMSOqr(\varphi)\leq s$ and $\sfFOqr(\varphi)\leq k$,
$t_1 \models \varphi \quad \mbox{if and only if} \quad t_2 \models \varphi$.
\end{corollary}
\begin{proof}
Since the predicates $\eda$ and $\era$ are not used
in the formula $\varphi \in \MSO(\sim,\prec,\dvSucc)$,
we can ignore them in $t_1$ and $t_2$ 
and view both $t_1$ and $t_2$ as ordered-data sets.
Our corollary follows immediately from Lemma~\ref{l: extended rep for data set}.
\end{proof}

\section{Automata for Ordered-data Tree}
\label{s: S-automata}

In this section we are going to introduce an automata model
for ordered-data trees and study its expressive power.

\begin{definition}
\label{d: S-automata}
An ordered-data tree automaton, in short ODTA, over the alphabet $\Sigma$ 
is a triplet $\S = \langle \T, \M, \Gamma_0\rangle$, where
$\T$ is a letter-to-letter non-deterministic transducer from $\Sigma\times\{\top,\bot,*\}^3$
to the output alphabet $\Gamma$;
$\M$ is an automaton on strings over the alphabet $2^{\Gamma}$;
and $\Gamma_0 \subseteq \Gamma$.
\end{definition}

An ordered-data tree $t$ is accepted by $\S$,
denoted by $t \in \L_{data}(\S)$, if there exists an ordered-data tree $t'$ over $\Gamma$
such that
\begin{itemize}\itemsep=0pt
\item
on input $\sfProfile(t)$, the transducer $\T$ outputs $t'$;
\item
the automaton $\M$ accepts the string $\V_{\Gamma}(t')$; and
\item
for every $a\in \Gamma_0$,
all the $a$-nodes in $t'$ have different data values.
\end{itemize}
We describe a few examples of ODTA
that accept the languages described in 
Examples~\ref{eg: two a-nodes},~\ref{eg: m data values}, 
\ref{eg: mod m data values} and~\ref{eg: max a-node}.

\begin{example}
\label{eg: automata two a-nodes}
An ODTA $\S^a = \langle \T,\M,\Gamma_0\rangle$
that accepts the language $\L_{data}^a$ in Example~\ref{eg: two a-nodes}
can be defined as follows.
The output alphabet of the transducer $\T$ is $\Gamma = \{\alpha,\beta,\gamma\}$.
On an input tree $t$,
the transducer $\T$ marks the nodes in $t$ as follows.
There is only one node marked with $\alpha$,
one node marked with $\beta$, and
the $\alpha$-node is an ancestor of $\beta$.
The automaton $\M$ accepts all the strings in which
the position labeled with $S\ni \beta$ is less than or equal to
the position labeled with $S'\ni \alpha$.
(These two positions can be equal, which means $S=S'$.)
Finally, $\Gamma_0 = \emptyset$.
\end{example}

\begin{example}
\label{eg: automata m data values}
An ODTA $\S^{S,m} = \langle \T,\M,\Gamma_0\rangle$ that accepts
the language $\L_{data}^{S,m}$ in Example~\ref{eg: m data values}
can be defined as follows.
The transducer $\T$ is an identity transducer.
The automaton $\M$ accepts all the strings in which
the symbol $S$ appears exactly $m$ times, and
$\Gamma_0 = \emptyset$.
\end{example}

\begin{example}
\label{eg: automata mod m data values}
An ODTA $\S^{S,\pmod{m}} = \langle \T,\M,\Gamma_0\rangle$ that accepts
the language $\L_{data}^{S,\pmod{m}}$ in Example~\ref{eg: mod m data values}
can be defined as follows.
The transducer $\T$ is an identity transducer.
The automaton $\M$ accepts a string in which
the number of appearances of the symbol $S$ is a multiple of $m$,
and $\Gamma_0 = \emptyset$.
\end{example}

\begin{example}
\label{eg: automata max a-nodes}
An ODTA $\S^{a\ast} = \langle \T,\M,\Gamma_0\rangle$
that accepts the language $\L_{data}^{a\ast}$ in Example~\ref{eg: max a-node}
can be defined as follows.
The output alphabet of the transducer $\T$ is $\Gamma = \{\alpha,\beta\}$.
The transducer $\T$ marks the nodes as follows.
A node is marked with $\alpha$ if and only if
it is an $a$-node and it has different data value from the one of its parent.
All the other nodes are marked with $\beta$.
The automaton $\M$ accepts a string $v$ if and only if
the last symbol in $v$ contains the symbol $\alpha$,
while $\Gamma_0 = \{\alpha\}$.
\end{example}

The following proposition states that ODTA languages 
are closed under union and intersection, but not under negation.
We would like to remark that being not closed under negation
is rather common for decidable models for data trees.
Often models that are closed under negation have undecidable 
non-emptiness/satisfiability problem.
\begin{proposition}
\label{p: boolean closure}
The class of languages accepted by ODTA
is closed under union and intersection,
but not under negation.
\end{proposition}
\begin{proof}
For closure under union and intersection,
let $\S_1 = \langle \T_1,\M_1,\Gamma_0^1\rangle$
and $\S_2 = \langle \T_2,\M_2,\Gamma_0^2\rangle$
be ODTA.
The union $\L_{data}(\S_1)\cup\L_{data}(\S_2)$
is accepted by an ODTA which
non-deterministically chooses to simulate either $\S_1$ or $\S_2$
on the input ordered-data tree.
The ODTA for the intersection $\L_{data}(\S_1)\cap\L_{data}(\S_2)$
can be obtained by the standard cross product between $\S_1$ and $\S_2$.

We now prove hat ODTA languages are not closed under negation.
Consider the negation of the language in Example~\ref{eg: two a-nodes},
whose equivalent ODTA $\S^a$ is presented in Example~\ref{eg: automata two a-nodes}.
Every tree $t \notin \L(\S^a)$ has the following property.
If $u,v$ are two $a$-nodes in $t$ and $u$ is an ancestor of $v$,
then $u \prec v$.

Now suppose to the contrary that 
there exists an ODTA $\S = \langle \T,\M,\Gamma_0\rangle$
that accepts the negation of $\L(\S^a)$.
Let $\Gamma$ be the output alphabet of $\T$.
Let $t \in \L(\S)$ be a data tree with $|\Gamma|+1$ nodes,
where each node is labelled with $a$
and has at most one child.
This implies that the data values in $t$ are all different
and appear in increasing order from the root node to the leaf node.

Let $t' \in \T(t)$.
Since $t$ has $|\Gamma|+1$ nodes,
and hence so does $t'$,
there are two nodes in $u$ and $v$ in $t'$ with the same label.
Let $t''$ be a data tree obtained from $t$ by swapping the data values between $u$ and $v$,
so $t'' \in \L(\S^a)$.
Since $\sfProfile(t)=\sfProfile(t'')$,
on input $\sfProfile(t'')$,
the transducer $\T$ can also output $t'$,
which means that $t'' \in \L(\S)$.
This contradicts the fact that $\L(\S)$ is the complement of $\L(\S^a)$.
This completes the proof of Proposition~\ref{p: boolean closure}.
\end{proof}

We should remark that in Section~\ref{s: undecidable}
we will discuss that extending ODTA with the complement of languages
of the form in Example~\ref{eg: automata two a-nodes}
will immediately yield undecidability.

Theorems~\ref{t: commutativity},~\ref{t: logic S-automata} and~\ref{t: decidability}
are the main results in this paper.
Theorem~\ref{t: commutativity} below provides the ODTA characterisation
of the logic $\EMSO^2(\eda,\era,\sim)$ and
its proof can be found in Subsection~\ref{ss: proof: t: commutativity}.

\begin{theorem}
\label{t: commutativity}
A language $\L_{data}$ is expressible with an $\EMSO^2(\eda,\era,\sim)$ formula
if and only if 
it is accepted by an ODTA $\S = \langle \T,\M,\Gamma_0\rangle$,
where $\L(\M)$ is a commutative language.
Moreover, the translation from $\EMSO^2(\eda,\era,\sim)$ formulas to ODTA
takes triple exponential time,
while from ODTA to $\EMSO^2(\eda,\era,\sim)$ formulas,
takes exponential time.  
\end{theorem}

Theorem~\ref{t: logic S-automata} below provides the logical characterisation of ODTA.
The proof can be found in Subsection~\ref{ss: proof: t: logic S-automata}.

\begin{theorem}
\label{t: logic S-automata}
A language $\L_{data}$ is accepted by an ODTA if and only if
it is expressible with a formula of the form:
$\exists X_1 \cdots \exists X_m \ \varphi \wedge \psi$,
where $\varphi$ is a formula from $\FO^2(\eda,\era,\sim)$,
and $\psi$ from $\FO(\sim,\prec,\dvSucc)$, both
extended with the unary predicates $X_1,\ldots,X_m$ and $a(\cdot)$.
Moreover, the translation from ODTA to formula is of polynomial time,
and from formula to ODTA is effective, but non-elementary.
\end{theorem}

Finally, we show that
the non-emptiness problem for ODTA is decidable in Theorem~\ref{t: decidability}.
The proof can be found in Subsection~\ref{ss: proof: t: decidability}.

% The decision procedure in Theorem~\ref{t: decidability} runs in 3-$\nexp$,
% while the decision procedure for $\EMSO^2(\eda,\era,\sim)$
% proposed in~\cite{BMSS09} also runs in 3-$\nexp$.
% However, we should remark that
% if we use our algorithm for the satisfiability problem of $\EMSO^2(\eda,\era,\sim)$
% via the translation in Theorem~\ref{t: commutativity},
% the complexity will jump to 5-$\nexp$,
% since there is a double exponential blow-up in the translation.

\begin{theorem}
\label{t: decidability}
The non-emptiness problem for ODTA is decidable in 3-$\nexp$.
\end{theorem}

The best lower bound known up to date is $\np$-hard.
See~\cite{FL-jacm,edt}.

\subsection{Proof of Theorem~\ref{t: commutativity}}
\label{ss: proof: t: commutativity}

In the proof we assume that the ordered-data trees are over the finite alphabet $\Sigma$.
We will need the following proposition which states that
every $\EMSO^2(\eda,\era,\sim)$ formula can be syntactically rewritten to
a normal form for $\EMSO^2(\eda,\era,\sim)$.
\begin{proposition}
\label{p: normal form EMSO2}{\em \cite[Proposition~3.8]{BMSS09}}
Every formula $\psi \in \EMSO^2(\eda,\era,\sim)$ can be rewritten into
a normal form of exponential size of the form: 
$\exists Y_1\cdots\exists Y_n \ \varphi$,
where $\varphi$ is a conjunction of formulae of the form:
\begin{itemize}
\item[\rm (N1)]
$\forall x \forall y \ (\alpha(x) \wedge \delta(x,y) \wedge \xi(x,y) \to \beta(y))$,
\item[\rm (N2)]
$\forall x \ (\root(x)  \to \alpha(y))$,
\item[\rm (N3)]
$\forall x \ (\fs(x)  \to \alpha(y))$,
\item[\rm (N4)]
$\forall x \ (\ls(x)  \to \alpha(y))$,
\item[\rm (N5)]
$\forall x \ (\leaf(x)  \to \alpha(y))$,
\item[\rm (N6)]
$\forall x \forall y \ (\alpha(x) \wedge \alpha(y) \wedge x\sim y \to x=y)$,
\item[\rm (N7)]
$\forall x \exists y \ (\alpha(x) \to \beta(y) \wedge x\sim y)$,
\end{itemize}
where $\alpha(x),\beta(x)$ is a conjunction of some unary predicates and its negations,
$\delta(x,y)$ is either $\eda(x,y)$ or $\era(x,y)$,
and
$\xi(x,y)$ is either $x\sim y$ or $x\nsim y$.
\end{proposition}

We should remark that if $\varphi$ is a conjunction 
of formulae of the forms (N1)--(N5) above,
then there exists a tree automaton $\A$ over the alphabet $\Sigma \times \{\top,\bot,*\}^3$
such that 
for every ordered-data tree $t$,
$$
t\models \Psi 
\quad \mbox{if and only if}\quad
\sfProfile(t) \ \mbox{is accepted by} \ \A.
$$
Such construction is straightforward from the classical
automata theory. See, for example,~\cite{thomas-handbook}.
We divide the proof of Theorem~\ref{t: commutativity} into 
Lemmas~\ref{l: commutative capture fo2} and~\ref{l: fo2 capture commutative}
below.

\begin{lemma}
\label{l: commutative capture fo2}
For every formula $\Psi \in \EMSO^2(\eda,\era,\sim)$,
there exists an ODTA $\S_{\Psi} = \langle \T,\M,\Gamma_0\rangle$
such that $\L_{data}(\Psi)=\L_{data}(\S_{\Psi})$
and $\L(\M)$ is commutative.
Moreover, the construction of $\S_{\Psi}$ is effective and 
takes triple exponential time in the size of the formula $\Psi$.
\end{lemma}
\begin{proof}
Applying Proposition~\ref{p: normal form EMSO2},
we can rewrite the formula $\Psi$ in its normal form $\exists Y_1\cdots\exists Y_n \Psi'$.
Furthermore, we can rewrite the formula $\Psi$ into the form
$\exists X_1\cdots\exists X_m \ \varphi$,
where $m=2^n$, and $\varphi$ is a conjunction of formulas of the form:
\begin{itemize}
\item[\rm (N0$'$)]
$X_1,\ldots,X_m$ are pairwise disjoint, and 
$\bigwedge_{a\in \Sigma} \forall x (a(x) \to \alpha'(x))$.
\item[\rm (N1$'$)]
$\forall x \forall y \ (\alpha'(x) \wedge \delta(x,y) \wedge \xi(x,y) \to \beta'(y))$,
\item[\rm (N2$'$)]
$\forall x \ (\root(x)  \to \alpha'(y))$,
\item[\rm (N3$'$)]
$\forall x \ (\fs(x)  \to \alpha'(y))$,
\item[\rm (N4$'$)]
$\forall x \ (\ls(x)  \to \alpha'(y))$,
\item[\rm (N5$'$)]
$\forall x \ (\leaf(x)  \to \alpha'(y))$,
\item[\rm (N6$'$)]
$\forall x \forall y \ (\alpha'(x) \wedge \alpha'(y) \wedge x\sim y \to x=y)$,
\item[\rm (N7$'$)]
$\forall x \exists y \ (\alpha'(x) \to \beta'(y) \wedge x\sim y)$,
\end{itemize}
where $\alpha'(x),\beta'(x)$ are disjunctions of
some of the $X_i$'s, and
$\delta(x,y)$ and $\xi(x,y)$ are the same above.
Intuitively, the unary predicates $X_1,\ldots,X_m$ corresponds to subsets of $\{Y_1,\ldots,Y_n\}$.

The ODTA $\S_{\Psi} = \langle \T,\M,\Gamma_0\rangle$
is defined as follows.
\begin{itemize}
\item 
The transducer $\T$ checks whether the formulas (N0$'$)--(N5$'$) 
are satisfied, with the output alphabet $\Gamma = \{X_1,\ldots,X_m\}$
where a node is labeled with $X_i$ if and only if it belongs to $X_i$.
\\
The construction of such transducer is straightforward, thus, omitted.
See, for example,~\cite{thomas-handbook}.
\item 
$\Gamma_0$ consists of the $X_i$'s,
where there exists $A \subseteq \{X_1,\ldots,X_m\}$
and $X_i \in A$ and a formula of the form (N6$'$)
$$
\forall x \forall y \ 
(\bigvee_{X_j \in A} X_j(x) \wedge \bigvee_{X_j \in A} X_j(y) \wedge x\sim y \to x=y),
$$
in $\varphi$.
\item 
the automaton $\M$ accepts the language $(2^{\{X_1,\ldots,X_m\}}-(\P_1\cup\P_2))^{\ast}$,
where 
\begin{eqnarray*}
\P_1 & := & 
\left\{
\begin{array}{l}
S 
\left|
\begin{array}{l} \mbox{there exists a formula}
\\
\qquad
\forall x \exists y \ (\bigvee_{X \in A} X(x) \to \bigvee_{X \in B} X(y) \wedge x\sim y)
\\ 
\mbox{in} \ \varphi \ \mbox{such that} 
     \ S \cap A \neq \emptyset \ \mbox{but} \ S \cap B = \emptyset
\end{array}
\right\}
\end{array}
\right.
\\
\P_2 & := & 
\left\{
\begin{array}{l}
S 
\left|
\begin{array}{l} \mbox{there exists a formula}
\\ 
\qquad
\forall x \forall y \ (\bigvee_{X \in A} X(x) \wedge \bigvee_{X \in A} X(y) \wedge x\sim y \to x=y) 
\\
 \mbox{in} \ \varphi \ \mbox{such that} \ |S \cap A| \geq 2
\end{array}
\right\}
\end{array}
\right.
\end{eqnarray*}
\end{itemize}
That $\L(\M)$ is commutative is trivial.
That $\S$ accepts precisely the language $\L_{data}(\Psi)$
can be deduced from the following.
\begin{itemize}
\item
That $\T$ ensures that formulas N0$'$--N5$'$ are satisfied.
\item
That $\Gamma_0$ contains precisely the symbols $X_i$'s
where all $X_i$-nodes are supposed to contain different data values.
\item
That for every ordered-data tree $t$,
$$
t \models \forall x \exists y \ (\bigvee_{X \in A} X(x) \to \bigvee_{X \in B} X(y) \wedge x\sim y)
$$ 
if and only if
$[S]_t = \emptyset \ \mbox{for all} \ S \ \mbox{such that} \ 
S \cap A \neq \emptyset \ \mbox{but} \ S \cap B = \emptyset$.
\item
That for every ordered-data tree $t$,
$$
t\models 
\forall x \forall y \ 
(\bigvee_{X \in A} X(x) \wedge \bigvee_{X \in A} X(y) \wedge x\sim y \to x=y)
$$
if and only if
\begin{itemize}
\item
$[S]_t = \emptyset$ for all $S$ such that $|S \cap A| \geq 2$; and
\item
for all $X \in A$, 
$t\models \forall x \forall y \ (X(x) \wedge X(y) \wedge x\sim y \to x=y)$,
which is captured by the condition imposed by $\Gamma_0$.
\end{itemize}
\end{itemize}
The analysis of the complexity is as follows.
The first step, applying Proposition~\ref{p: normal form EMSO2},
induces an exponential blow-up in the size of the input.
The second step to construct the formula $\exists X_1 \cdots \exists X_m \ \varphi$
takes exponential time in $n$,
and $n$ is exponential in the size of the input.
The construction of $\T$ takes polynomial time in the size of $\varphi$,
since (N0$'$)--(N5$'$) are already in the ``automata transition'' format.
The construction of $\Gamma_0$ takes polynomial time in $m$,
while the construction of $\M$ induces another exponential blow-up in $m$.
Altogether the complexity of our constructing $\S_{\Psi}$
is triple exponential time in the size of $\Psi$. 
This concludes the proof of Lemma~\ref{l: commutative capture fo2}.
\end{proof}

For the complexity analysis in Lemma~\ref{l: fo2 capture commutative},
we assume that a commutative automaton $\M$
is given as a set of vectors (in binary format) indicating its Parikh images.
That is, $\M$ is given as a set $I = \{(\bar{u}_1,\bar{v}_{1,1},\ldots,\bar{v}_{1,\ell}),\ldots,
(\bar{u}_n,\bar{v}_{n,1},\ldots,\bar{v}_{n,\ell})\}$, where
$$
\bigcup_{(\bar{u},\bar{v}_1,\ldots,\bar{v}_{\ell}) \in I} 
\L(\bar{u},\bar{v}_1,\ldots,\bar{v}_{\ell})
= 
\L(\M),
$$
and each number in the vectors in $I$ is written in the standard binary form.
 
\begin{lemma}
\label{l: fo2 capture commutative}
For every ODTA $\S = \langle \T,\M,\Gamma_0\rangle$,
where $\L(\M)$ is a commutative language,
there exists a formula $\varphi \in \EMSO^2(\eda,\era,\sim)$
such that $\L_{data}(\varphi)=\L_{data}(\S)$. 
Moreover, the construction of $\varphi$ takes exponential time in the size of $\S$.
\end{lemma}
\begin{proof}
Let $Q_{\sT} = \{q_0,\ldots,q_m\}$ and $\Gamma = \{\alpha_1,\ldots,\alpha_k\}$ 
be the set of states and the output alphabet of the transducer $\T$, respectively.
Let $\ell = 2^{|\Gamma|}-1$.

By Theorem~\ref{t: commutative regular},
$\L(\M)$ is a finite union of periodic languages.
Let $I$ be the finite set of $(\ell+1)$-tuple of $\bbN^{\ell}$-vectors
such that
$$
\bigcup_{(\bar{u},\bar{v}_1,\ldots,\bar{v}_{\ell}) \in I} 
\L(\bar{u},\bar{v}_1,\ldots,\bar{v}_{\ell})
= 
\L(\M).
$$
Let $I = \{(\bar{u}_1,\bar{v}_{1,1},\ldots,\bar{v}_{1,\ell}),\ldots,
(\bar{u}_n,\bar{v}_{n,1},\ldots,\bar{v}_{n,\ell})\}$
and $S_1,\ldots,S_{\ell}$ be the enumeration of non-empty subsets of $\Gamma$.
First, for $(\bar{u},\bar{v}_1,\ldots,\bar{v}_{\ell}) \in I$, 
we construct an $\EMSO^2(\eda,\era,\sim)$ formula 
$\Psi_{(\bar{u},\bar{v}_1,\ldots,\bar{v}_{\ell})}$
where
\begin{eqnarray*}
t \in \Psi_{(\bar{u},\bar{v}_1,\ldots,\bar{v}_{\ell})}
& \mbox{if and only if} &  
\left[
\begin{array}{l}
 \mbox{there exists}\  h_1,\ldots,h_{\ell}\geq 0 
\ \mbox{such that} 
\\
(|[S_1]_t|,\ldots,|[S_{\ell}]_t|) 
=
\bar{u} + h_1 \bar{v}_1 + \cdots + h_{\ell} \bar{v}_{\ell}
\end{array}
\right]
\end{eqnarray*}
We denote by $v_i$ the non-zero entry of $\bar{v}_i$.
This formula $\Psi_{(\bar{u},\bar{v}_1,\ldots,\bar{v}_{\ell})}$ is as follows.
\begin{eqnarray*}
& & \exists W_{1,1} \cdots W_{1,u_1} \ \cdots\cdots \ \exists W_{\ell,1} \cdots W_{\ell,u_{\ell}}
\\
& & \quad\quad \exists X_{1,0} \cdots X_{1,v_1-1} \ \exists Y_{1,0} \cdots Y_{1,v_1-1} \ Z_1
\\
& & \quad\quad\quad\quad\quad \ddots
\\
& & \quad\quad\quad\quad\quad
\exists X_{\ell,0} \cdots X_{\ell,v_\ell-1} \ \exists Y_{\ell,0} \cdots Y_{\ell,v_\ell-1} \ Z_\ell
\\
& & \quad\quad\quad\quad\quad\quad\quad\quad
\bigwedge_{i} W_{i,1},\ldots,W_{i,u_i} \cap Z_i = \emptyset
\\
& & \quad\quad\quad\quad\quad\quad\quad\quad\wedge \
\bigwedge_{i} \varphi_{|[S_i]|=u_i}(W_{i,1},\ldots,W_{i,u_i})
\\
& & \quad\quad\quad\quad\quad\quad\quad\quad\wedge \
\bigwedge_{i} 
\varphi_{|[S_i]|\equiv v_i\pmod{m}}(X_{i,0},\ldots,X_{i,v_i-1},Y_{i,0},\ldots,Y_{i,v_i-1},Z_i)
\end{eqnarray*}
where $\varphi_{S_i,u_i}(W_{i,1},\ldots,W_{i,u_i})$
and $\varphi_{S_i,\pmod{v_i}}(X_{i,0},\ldots,X_{i,v_i-1},Y_{i,0},\ldots,Y_{i,v_i-1},Z_i)$
are the formulas for the languages $\L_{data}^{S_i,u_i}$
and $\L_{data}^{S_i,\pmod{u_i}}$
in Examples~\ref{eg: m data values} and~\ref{eg: mod m data values}, respectively.

The desired formula $\varphi$ is:
\begin{eqnarray*}
& & 
\exists X_{q_0} \cdots \exists X_{q_m} 
\ \exists X_{\alpha_1} \cdots \exists X_{\alpha_k}
\ \exists \overline{X}_{(\bar{u_1},\bar{v}_{1,1},\ldots,\bar{v}_{1,\ell})}
\ \cdots \ \exists \overline{X}_{(\bar{u_n},\bar{v}_{n,1},\ldots,\bar{v}_{n,\ell})}
\\
& &
\qquad\qquad\qquad\qquad\qquad\varphi_{\Gamma_0} \wedge \varphi_{\sT} \wedge 
\bigvee_{(\bar{u},\bar{v}_{1},\ldots,\bar{v}_{\ell}) \in I} 
\varphi_{(\bar{u},\bar{v}_{1},\ldots,\bar{v}_{\ell})}
\end{eqnarray*}
where
\begin{itemize}
\item
the formula $\varphi_{\Gamma_0}$ expresses the fact that
the data values found under nodes labeled with a symbol from $\Gamma_0$ 
are all different;
\item
the unary predicates $X_{q_0}, \ldots, X_{q_m},X_{\alpha_1}, \ldots,X_{\alpha_k}$
are supposed to represent the states and the output alphabets of $\T$, respectively;
\item
the formula $\varphi_{\sT}$ expresses the behaviour of the transducer $\T$ -- that is,
a tree satisfies $\varphi_{\sT}$ in which for every node $u \in \Domain(t)$,
$X_{q_i}(u)$ and $X_{\alpha_j}(u)$ holds, if there exists an accepting run of $\T$ on $t$
in which the node $u$ is labeled with $q_i$ and output $\alpha_j$;
\item
the predicates $\overline{X}_{(\bar{u_i},\bar{v}_{i,1},\ldots,\bar{v}_{i,\ell})}$'s and
the formulas $\varphi_{(\bar{u_i},\bar{v}_{i,1},\ldots,\bar{v}_{i,\ell})}$'s
are as in the formula $\Psi_{(\bar{u},\bar{v}_1,\ldots,\bar{v}_{\ell})}$ defined above. 
\end{itemize}
The analysis of the complexity is as follows.
The size of the formula $\varphi_{S_i,u_i}$ and $\varphi_{S_i,\pmod{v_i}}$
are exponential in the size of $S_i,u_i,v_i$.
Hence, the construction of $\Psi_{(\bar{u},\bar{v}_1,\ldots,\bar{v}_{\ell})}$
takes exponential time in the size of $(\bar{u},\bar{v}_1,\ldots,\bar{v}_{\ell})$.
The construction of $\varphi_{\Gamma_0}$ and $\varphi_{\sT}$ takes polynomial time
in the size of $\Gamma_0$ and $\T$, respectively.
Hence, the total time to construct the formula $\varphi$
is exponential in the size of $\S$.
This completes the proof of the lemma.
\end{proof}

\subsection{Proof of Theorem~\ref{t: logic S-automata}}
\label{ss: proof: t: logic S-automata}

In this subsection for every ordered-data tree $t$,
we assume that the data values in $t$
are precisely the natural numbers in the range $[1 .. m]$,
for a positive integer $m \geq 1$.
That is, if $d_1< d_2 < \cdots < d_m$
are the data values in $t$, then
$d_1 =1$, $d_2 = 2$, $\ldots$, $d_m=m$. 

We start with the following lemma.
\begin{lemma}
\label{l: aut for fo}
Let $\psi \in \FO(\sim,\prec)$ be of quantifier rank $k$.
Let $\Gamma= \{a_1,\ldots,a_{\ell}\}$ be the set of unary predicates
used in $\psi$.
There exists a finite state automaton $C$ 
over the alphabet $\Gamma \cup (2^{\Gamma}\times \F_{\Gamma,k})$ such that
the following holds.
\begin{itemize}\itemsep=0pt
\item
The automaton $C$ accepts words of the form
$$
\overbrace{a_1\cdots a_1}^{f_1(a_1)} \cdots \overbrace{a_{\ell}\cdots a_{\ell}}^{f_1(a_{\ell})}
\ (S_1,f_1) \
\cdots\cdots \
\overbrace{a_1\cdots a_1}^{f_m(a_1)} \cdots \overbrace{a_{\ell}\cdots a_{\ell}}^{f_m(a_{\ell})}
\ (S_m,f_m),
$$
where each $S_i = \{a \mid f_i(a)\geq 1\}$.
\item
For every ordered-data tree $t \models \psi$, 
if $\V^{(k)} = (S_1,f_1),\ldots,(S_m,f_m)$,
then there exists a word in $\L(C)$ of the form
$$
\overbrace{a_1\cdots a_1}^{f_1(a_1)} \cdots \overbrace{a_{\ell}\cdots a_{\ell}}^{f_1(a_{\ell})}
\ (S_1,f_1) \
\cdots\cdots \
\overbrace{a_1\cdots a_1}^{f_m(a_1)} \cdots \overbrace{a_{\ell}\cdots a_{\ell}}^{f_m(a_{\ell})}
\ (S_m,f_m)
$$
\item
For every word $w \in \L(C)$,
if $w$ is
$$
\overbrace{a_1\cdots a_1}^{f_1(a_1)} \cdots \overbrace{a_{\ell}\cdots a_{\ell}}^{f_1(a_{\ell})}
\ (S_1,f_1) \
\cdots\cdots \
\overbrace{a_1\cdots a_1}^{f_m(a_1)} \cdots \overbrace{a_{\ell}\cdots a_{\ell}}^{f_m(a_{\ell})}
\ (S_m,f_m)
$$
then there exists a tree $t\models \psi$,
where $\V^{(k)}(t) = (S_1,f_1)\cdots (s_m,f_m)$.
\end{itemize}
\end{lemma}
\begin{proof}
Let $\psi \in \FO(\sim,\prec)$ be of quantifier rank $k$.
Let $\Gamma= \{a_1,\ldots,a_{\ell}\}$ be the set of unary predicates
used in $\varphi$.
We define the following sentence $\overline{\psi} \in \FO(<)$ (that is, over strings)
inductively from $\psi$ as follows.
\begin{itemize}\itemsep=0pt
\item 
If $\psi$ is $Qx \; \xi$, where $Q \in\{\forall,\exists\}$,
then $\overline{\psi}$ is 
$$
Qx\; \bigvee_{a\in \Gamma} a(x) \to \overline{\xi}.
$$
\item
If $\psi$ is $x=y$, then
$\overline{\psi}$ is also $x=y$.
\item
If $\psi$ is $x\sim y$, then
$\overline{\psi}$ states ``there is no position in between $x$ and $y$
labeled with any symbol from $2^{\Gamma}\times\F_{\Gamma,k}$.''
\item
If $\psi$ is $x\prec y$, then
$\overline{\psi}$ states ``there is at least one position in between $x$ and $y$
labeled with a symbol from $2^{\Gamma}\times\F_{\Gamma,k}$.''
\end{itemize}

We have the following claim.
\begin{claim}
\begin{enumerate}[(1)]\itemsep=0pt
\item
For every ordered-data tree $t \models \psi$, 
if $\V^{(k)} = (S_1,f_1),\ldots,(S_m,f_m)$,
then there exists a word $w \models \overline{\psi}$ of the form
$$
\overbrace{a_1\cdots a_1}^{f_1(a_1)} \cdots \overbrace{a_{\ell}\cdots a_{\ell}}^{f_1(a_{\ell})}
\ (S_1,f_1) \
\cdots\cdots \
\overbrace{a_1\cdots a_1}^{f_m(a_1)} \cdots \overbrace{a_{\ell}\cdots a_{\ell}}^{f_m(a_{\ell})}
\ (S_m,f_m)
$$
\item
For every word $w \models \overline{\psi}$,
if $w$ is
$$
\overbrace{a_1\cdots a_1}^{f_1(a_1)} \cdots \overbrace{a_{\ell}\cdots a_{\ell}}^{f_1(a_{\ell})}
\ (S_1,f_1) \
\cdots\cdots \
\overbrace{a_1\cdots a_1}^{f_m(a_1)} \cdots \overbrace{a_{\ell}\cdots a_{\ell}}^{f_m(a_{\ell})}
\ (S_m,f_m)
$$
then there exists a tree $t\models \psi$,
where $\V^{(k)}(t) = (S_1,f_1)\cdots (s_m,f_m)$.
\end{enumerate}
\end{claim}
\begin{proof}
We first prove item~(1).
Let $t$ be an ordered-data tree over the alphabet $\Gamma$ and
let $\V^k(t) = (S_1,f_1) \cdots (S_m,f_m)$ be its
$k$-extended string representation of data values in $t$.
Let $t'$ be the following data string
$$
\overbrace{{a_1 \choose 1}\cdots {a_1 \choose 1}}^{f_1(a_1)}
\cdots 
\overbrace{{a_{\ell} \choose 1}\cdots {a_{\ell} \choose 1}}^{f_1(a_{\ell})}
\cdots
\cdots
\cdots
\cdots
\overbrace{{a_1 \choose m}\cdots {a_1 \choose m}}^{f_m(a_1)}
\cdots 
\overbrace{{a_{\ell} \choose m}\cdots {a_{\ell} \choose m}}^{f_m(a_{\ell})}
$$
When $t'$ is viewed as a data tree\footnote{That is, a data string is a data tree in which
each node has at most one child.}, $\V^{(k)}_{\Gamma}(t)=\V^{(k)}(t')$.
Hence, by Corollary~\ref{cor: extended rep for mso},
$$
t\models \psi \qquad \mbox{if and only if}\qquad 
t' \models \psi.
$$ 
By straightforward induction on $\overline{\psi}$,
we can show that
for every $t' \models \psi$ of the form
$$
\overbrace{{a_1 \choose 1}\cdots {a_1 \choose 1}}^{f_1(a_1)}
\cdots 
\overbrace{{a_{\ell} \choose 1}\cdots {a_{\ell} \choose 1}}^{f_1(a_{\ell})}
\cdots
\cdots
\cdots
\cdots
\overbrace{{a_1 \choose m}\cdots {a_1 \choose m}}^{f_m(a_1)}
\cdots 
\overbrace{{a_{\ell} \choose m}\cdots {a_{\ell} \choose m}}^{f_m(a_{\ell})}
$$
there exists a word $w \models \overline{\psi}$ of the form
$$
\overbrace{a_1\cdots a_1}^{f_1(a_1)} \cdots \overbrace{a_{\ell}\cdots a_{\ell}}^{f_1(a_{\ell})}
\ (S_1,f_1) \
\cdots\cdots \
\overbrace{a_1\cdots a_1}^{f_m(a_1)} \cdots \overbrace{a_{\ell}\cdots a_{\ell}}^{f_m(a_{\ell})}
\ (S_m,f_m)
$$

Similarly, to prove~(2), we can prove by straightforward induction on $\overline{\psi}$ that
for every word $w \models \overline{\psi}$ of the form
$$
\overbrace{a_1\cdots a_1}^{f_1(a_1)} \cdots \overbrace{a_{\ell}\cdots a_{\ell}}^{f_1(a_{\ell})}
\ (S_1,f_1) \
\cdots\cdots \
\overbrace{a_1\cdots a_1}^{f_m(a_1)} \cdots \overbrace{a_{\ell}\cdots a_{\ell}}^{f_m(a_{\ell})}
\ (S_m,f_m),
$$
there exists a tree $t\models \psi$ of the form
$$
\overbrace{{a_1 \choose 1}\cdots {a_1 \choose 1}}^{f_1(a_1)}
\cdots 
\overbrace{{a_{\ell} \choose 1}\cdots {a_{\ell} \choose 1}}^{f_1(a_{\ell})}
\cdots
\cdots
\cdots
\cdots
\overbrace{{a_1 \choose m}\cdots {a_1 \choose m}}^{f_m(a_1)}
\cdots 
\overbrace{{a_{\ell} \choose m}\cdots {a_{\ell} \choose m}}^{f_m(a_{\ell})}
$$
This completes the proof of our claim.
\end{proof}

Let $C$ be an automaton over the alphabet $\Gamma \cup (2^{\Gamma}\times \F_{\Gamma,k})$
that expresses the formula $\overline{\psi}$ and that
it accepts only words of the form
$$
\overbrace{a_1\cdots a_1}^{f_1(a_1)} \cdots \overbrace{a_{\ell}\cdots a_{\ell}}^{f_1(a_{\ell})}
\ (S_1,f_1) \
\cdots\cdots \
\overbrace{a_1\cdots a_1}^{f_m(a_1)} \cdots \overbrace{a_{\ell}\cdots a_{\ell}}^{f_m(a_{\ell})}
\ (S_m,f_m),
$$
where each $S_i = \{a \mid f_i(a)\geq 1\}$.
The construction of $C$ from the formula $\overline{\psi}$ is rather standard,
but non-elementary.
See, for example,~\cite{thomas-handbook}.
That the automaton $C$ is the desired automaton is immediate.
This completes our proof of Lemma~\ref{l: aut for fo}.
\end{proof}

\begin{lemma}
\label{l: aut for fo string rep}
Let $\psi \in \FO(\sim,\prec)$ be of quantifier rank $k$.
Let $\Gamma= \{a_1,\ldots,a_{\ell}\}$ be the set of unary predicates
used in $\psi$.
There exists a finite state automaton $\M$ 
over the alphabet $2^{\Gamma}\times \F_{\Gamma,k}$ such that
$\L(\M) = \{\V^{(k)}_{\Gamma,k}(t) \mid t \models \psi\}$.
\end{lemma}
\begin{proof}
Let $C$ be the automaton obtained by applying Lemma~\ref{l: aut for fo} on the formula $\psi$.
Then let $\M$ be the automaton obtained from $C$,
where every symbol from $\Gamma$ is projected to empty string.
The automaton $\M$ is the desired automaton,
and this completes our proof of Lemma~\ref{l: aut for fo string rep}.
\end{proof}

Now we are ready to prove Theorem~\ref{t: logic S-automata}.
We start with the ``if'' direction.
Let $\Psi$ be a formula of the form:
$$
\exists Y_1 \cdots \exists Y_n \ \varphi \wedge \psi,
$$
$\varphi$ is a formula from $\FO^2(\eda,\era,\sim)$
and
$\psi$ from $\FO(\sim,\prec)$,
both extended with the unary predicates $Y_1,\ldots,Y_n$.

By Proposition~\ref{p: normal form EMSO2},
we can rewrite (with additional unary predicates)
the formula $\varphi$ into a conjunction of formulae of the form N1--N7 
as stated in Proposition~\ref{p: normal form EMSO2}.
Then we further rewrite it into the form
$$
\exists X_1 \cdots \exists X_m \ \varphi' \wedge \psi',
$$
where $m=2^n$ and
$\varphi$ is a formula from $\FO^2(\eda,\era,\sim)$
and
$\psi$ from $\FO(\sim,\prec)$,
both extended with the unary predicates $X_1,\ldots,X_m$,
and that the formula $\varphi'$ is conjunction of the form:
\begin{itemize}\itemsep=0pt
\item[(N0$'$)]
a formula $\xi$ that states that $X_1,\ldots,X_m$ are pairwise disjoint
and that 
$$
\bigwedge_{a\in\Sigma} \forall x \ (a(x)\to \alpha(x)),
$$
\item[(N1$'$)]
$\forall x \forall y \ (\alpha(x) \wedge \delta(x,y) \wedge \xi(x,y) \to \beta(y))$,
\item[(N2$'$)]
$\forall x \ (\root(x)  \to \alpha(y))$,
\item[(N3$'$)]
$\forall x \ (\fs(x)  \to \alpha(y))$,
\item[(N4$'$)]
$\forall x \ (\ls(x)  \to \alpha(y))$,
\item[(N5$'$)]
$\forall x \ (\leaf(x)  \to \alpha(y))$,
\item[(N6$'$)]
$\forall x \forall y \ (\alpha(x) \wedge \alpha(y) \wedge x\sim y \to x=y)$,
\item[(N7$'$)]
$\forall x \exists y \ (\alpha(x) \to \beta(y) \wedge x\sim y)$,
\end{itemize}
where $\alpha(x),\beta(x)$ are disjunctions of some of the unary predicates $X_1,\ldots,X_m$.

We will describe the ODTA $\S = \langle \T,\M,\Gamma_0\rangle$
for the formula $\Psi$, where
the transducer $\T$ expresses the formula N0$'$--N5$'$ with the output alphabet $\Gamma=\{X_1,\ldots,X_m\}$,
the automaton $\M$ expresses the formula N6$'$, N7$'$ and $\psi'$,
and $\Gamma_0$ is the set of symbols that appear in formula N6$'$.
Formally, it is defined as follows.
\begin{itemize}\itemsep=0pt
\item
The output alphabet of $\T$ is $\Gamma = \{X_1,\ldots,X_m\}$.
\item
The transducer expresses the formula N0$'$--N5$'$ above.
In particular, the input and output symbols of each node must satisfy the formula N0$'$.

This step take polynomial time, since the formula N0$'$--N5$'$
is already in the transition format.
\item
The set $\Gamma_0 = \{X_i \mid X_i \ \mbox{appears in} \ \mbox{N6}'\}$.

This step takes polynomial time.
\item
The automaton $\M$ expresses the formulas N6$'$, N7$'$ and $\psi'$,
obtained by applying Lemma~\ref{l: aut for fo string rep}.

This step is constructive, but non-elementary due to
the conversion from a formula to its finite state automaton.
\end{itemize}
It is straightforward to show that $\L_{data}(\S) = \{t \mid t\models\Psi\}$.

Now we prove the ``only if'' direction.
Let $\L = \L_{data}(\S)$, where $\S = \langle \T,\M,\Gamma_0\rangle$, and
\begin{itemize}
\item
$Q = \{q_1,\ldots,q_n\}$ be the states of $\T$;
\item
$P = \{p_1,\ldots,p_s\}$ be the states of $\M$,
and $p_1$ is the initial state of $\M$;
\item
$\Gamma = \{\alpha_1,\ldots,\alpha_{\ell}\}$ be the output alphabet of $\T$.
\end{itemize}
We denote by $\Sigma$ the input alphabet of $\T$.

The desired formula for $\L$ is of the form:
$$
\exists X_{q_1} \cdots \exists X_{q_n} \
\exists X_{\alpha_1} \cdots \exists X_{\alpha_{\ell}} \
\exists X_{p_1} \cdots \exists X_{p_s} \
\Psi_{\sT} \wedge \Psi_{\sM} \wedge \Psi_{\Gamma_0}
$$
where
\begin{itemize}
\item
the unary predicates 
$X_{q_1}, \ldots,X_{q_n}, X_{\alpha_1}, \ldots, X_{\alpha_{\ell}},X_{p_1}, \ldots, X_{p_s}$
are supposed to represent the states, the output alphabets of $\T$,
and the states of $\M$, respectively;
\item
the formula $\Psi_{\sT}$ expresses the behaviour of the transducer $\T$ -- that is,
a tree satisfies $\Psi_{\sT}$ in which for every node $u \in \Domain(t)$,
$X_{q_i}(u)$ and $X_{\alpha_j}(u)$ holds, if there exists an accepting run of $\T$ on $t$
in which the node $u$ is labeled with $q_i$ and output $\alpha_j$;
\item
the formula $\Psi_{\sM}$ expresses the behaviour of the automaton $\M$;
\item
the formula $\Psi_{\Gamma_0}$ expresses the property that
for every $\alpha_i \in \Gamma_0$, all the nodes belonging to $X_{\alpha_i}$ 
contain different data values, which is
$$
\bigwedge_{\alpha \in \Gamma_0}
\forall x \forall y (X_{\alpha}(x) \wedge X_{\alpha}(y) \wedge x\sim y \to x=y).
$$
\end{itemize}
The construction of the formula $\Psi_{\sT}$
is rather standard, thus, omitted.
We will show the construction of the formula $\Psi_{\sM}$.
Let $\Phi_{[S]}(x)$ denote the following formula
$$
\bigvee_{\alpha_i \in S} X_{\alpha_i}(x) 
\wedge
\bigwedge_{\alpha_i \in S} \exists y ( X_{\alpha_i}(x)  \wedge x\sim y)
\wedge
\bigwedge_{\alpha_j \notin S} \forall y (  X_{\alpha_j}(y) \to  x\nsim y),
$$
which states that the data value on the node $x$ belongs to $[S]$.
The formula $\Psi_{\sM}$ expresses the following properties.
\begin{itemize}
\item
That the node contains the minimal data value belongs to $X_{p_1}$.
Formally, it can be written as follows.
$$
\forall x (\forall y x \prec y \vee x \sim y \to X_{p_1}(x))
$$
\item
That the transition $\mu$ of $\M$ must be ``respected.''
Formally, it can be written as follows.
$$
\bigwedge_{(p_i,S,p_j) \in \mu}
\Big(
\forall x \forall y 
(X_{p_i}(x) \wedge \Psi_{[S]}(x) \wedge x \dvSucc y \to X_{p_j}(y))
\Big),
$$
where $x \dvSucc\ y$ stands for $x \prec y \wedge \forall z (\neg (x \prec z \wedge  z \prec y))$.
\item
That the node contains the maximal data value belongs to one of the final states of $\M$,
denoted by $F$.
Formally, it can be written as follows.
$$
\forall x (\forall y \ (y \prec x \vee y \sim x) \to \bigvee_{p_i \in F} X_{p_i}(x)).
$$
\end{itemize}
That the construction takes polynomial time is straightforward.
This completes our proof of Theorem~\ref{t: logic S-automata}.

\subsection{Proof of Theorem~\ref{t: decidability}}
\label{ss: proof: t: decidability}

The proof of Theorem~\ref{t: decidability} consists of two main steps.
\begin{enumerate}[(1)]\itemsep=0pt
\item
We prove that for each ODTA $\S$,
if $\L_{data}(\S)\neq\emptyset$,
then $\L_{data}(\S)$ contains a data tree with ``small model property''
(Lemma~\ref{l: small model}).
\item
We describe a procedure, 
that given an ODTA $\S$,
checks whether $\L(\S)$ contains a data tree with ``small model property,''
by converting the ODTA $\S$ into an APC $(\A,\xi)$.
Since the non-emptiness of APC is decidable,
Theorem~\ref{t: decidability} follows immediately.
\end{enumerate}
The first step (Lemma~\ref{l: small model})
is adapted from the proof of~\cite[Proposition~3.10]{BMSS09}.
It is in the second step our proof differs from~\cite[Proposition~3.10]{BMSS09}
The decision procedure in~\cite{BMSS09}
relies on intricate counting argument of the so called dog and sheep symbols
(see~\cite[page~36]{BMSS09})
and it seems that it cannot be generalised to the case of ODTA.
On the other hand, our decision procedure relies mainly on 
Proposition~\ref{p: aut-to-presburger},
Lemma~\ref{l: canonical lemma}
and counting the cardinality of each $[S]$.

We need a few terminologies.
A set of nodes in a data tree $t$ is called connected,
if it is connected in the graph induced by $\eda$ and $\era$.
A {\em zone} in a data tree $t$ is a maximal connected set of nodes
with the same data value.
The {\em outdegree} of a zone $Z$ is the number of different zones to which
there is an edge (either $\eda$ or $\era$) from $Z$.

Let $\S = \langle \T,\M,\Gamma_0\rangle$
be an ODTA, where $\T$ is a transducer from $\Sigma$ to $\Gamma$.
Let $Q$ be the set of states of $\T$.
For a tree $t \in \L_{data}(\S)$,
its extended tree $\tilde{t}$ (with respect to the ODTA $\S$)
is a tree over the alphabet $\Sigma\times\{\top,\bot,*\}^3\times Q \times \Gamma$,
where
\begin{itemize}
\item
the projection of $\tilde{t}$ to $\Sigma\times\{\top,\bot,*\}^3$
is $\sfProfile(t)$;
\item
the projection of $\tilde{t}$ to $Q$ is an accepting run of $\T$ on $t$;
\item
the projection of $\tilde{t}$ to $\Gamma$ is an output of $\T$ on $t$.
\end{itemize}

The following Lemma is simply an adaptation of~\cite[Proposition~3.10]{BMSS09}
to the case of ODTA.
The proof is via cut-and-paste, where
given an ordered-data tree $t$ over the alphabet $\Sigma$
where $t$ has ``many'' zones in which the outdegree is ``large,''
we can cut some nodes in $t$ and paste it in another part of $t$
without affecting the set $V_t(a)$'s for each $a\in \Sigma$.
The aim of such cut-and-paste is to reduce the number of zones in $t$ with large outdegree.
We give the formal statement below.

\begin{lemma}
\label{l: small model}
{\em [Compare~\cite[Proposition~3.10]{BMSS09}]}
For every ODTA $\S = \langle \T,\M,\Gamma_0\rangle$ over the alphabet $\Sigma$,
if $\L_{data}(\S) \neq \emptyset$,
then there exists a data tree $t \in \L_{data}(\S)$
in which there are at most $K^{O(K^2)}$ zones with outdegree $\geq K ^{(K^3)}$,
where 
$K = O( |\Sigma|\cdot |Q|\cdot |\Gamma|)$ and $Q$ is the set of states of $\T$
and $\Gamma$ the output alphabet of $\T$.
\end{lemma}
\begin{proof}
Let $\S = \langle \T,\M,\Gamma_0\rangle$ be an ODTA
over the alphabet $\Sigma$, and $Q$ is the set of states of $\T$
and $\Gamma$ the output alphabet of $\T$.
Suppose that $t_0 \in \L_{data}(\S)$.
We will work on the extended tree $\tilde{t}_0$ of $t_0$.
The aim is to convert $\tilde{t}_0$ into 
another tree $\tilde{t}$ over the alphabet $\Sigma\times\{\top,\bot,*\}^3\times Q \times \Gamma$
such that
\begin{enumerate}
\item
the number of zones in $\tilde{t}$ with outdegree $\geq K ^{(K^3)}$ is bounded by $K^{O(K^2)}$,
\item
the $\{\top,\bot,*\}^3$ projection of $\tilde{t}$ is the profile of each node,
\item
the $Q$ projection of $\tilde{t}$ is an accepting run of $\T$
on the $\Sigma\times\{\top,\bot,*\}^3$ projection of $\tilde{t}$ 
and the output is its $\Gamma$ projection,
\item
for each $(a,(l,p,r),q,b)\in \Sigma\times\{\top,\bot,*\}^3\times Q \times \Gamma$
the set of data values found in the $(a,(l,p,r),q,b)$-nodes in $\tilde{t}_0$
is the same as the set of those found in $(a,(l,p,r),q,b)$-nodes in $\tilde{t}$,
\item
the $\Sigma$ projection of $\tilde{t}$ is accepted by $\S$.
\end{enumerate}
Intuitively, the tree $\tilde{t}$ is obtained via repeated applications of ``pumping lemma'' 
on both $\eda$- and $\era$-directions in the tree $t$.

Below we give a brief summary of the proof adapted from 
the proof of~\cite[Proposition~3.10]{BMSS09}.
We need the following terminologies, all of them are from~\cite{BMSS09}.
\begin{itemize}
\item
Two nodes in a tree are called {\em siblings}, 
if they have the same parent node.
\item
The set of all children of a node is called a {\em sibling group}.
\item
A contiguous sequence of siblings is called an {\em interval}.
\\
We write $[u,v]$ for an interval in which 
$u$ and $v$ are the left-most and right-most nodes, respectively, in the interval. 
\item
An interval $[u,v]$ is {\em complete}, if the following holds.
\begin{itemize}
\item
If a node $u'$ exists such that $\era(u',u)$,
then $u'\nsim u$.
\item
If a node $v'$ exists such that $\era(v,v')$,
then $u'\nsim u$.
\end{itemize}
\item
An interval is {\em pure}, if all of its nodes have the same data value.
\item
A pure interval with the data value $d$ is called a $d$-pure interval.
\item
If the parent of an interval (or, a sibling group) has data value $d$,
then it is called a {\em $d$-parent interval} (or a {\em $d$-parent sibling group}).
\item
A zone with the data value $d$ is called a {\em $d$-zone}.
\end{itemize}
The construction of $\tilde{t}$ from $\tilde{t}_0$ is as follows.
\begin{enumerate}
\item
Convert $\tilde{t}_0$ to another tree $\tilde{t}_1$ such that
\begin{itemize}
\item
for every data value $d \in V_{\tilde{t}_1}$
there are at most $O(K)$ complete $d$-pure intervals of size more than $O(K)$;
\item
$V_{\tilde{t}_1}(a,(l,p,r),q,b) = V_{\tilde{t}_0}(a,(l,p,r),q,b)$,
for every $(a,(l,p,r),q,b) \in \Sigma\times\{\top,\bot,*\}^3\times Q\times \Gamma$;
\item
$\tilde{t}_1$ is an extended tree of its $\Sigma$ projection w.r.t. $\S$.
\end{itemize}
This step is adapted from~\cite[Proposition~3.12]{BMSS09}.
The idea is to cut an interval (together with its subtree)
and paste it in another interval; and while doing so
the data values in the interval remain untouched.

\item
Convert $\tilde{t}_1$ to another tree $\tilde{t}_2$ such that
\begin{itemize}
\item
for every data value $d \in V_{\tilde{t}_2}$
there are at most $O(K)$ $d$-parent sibling group
with more than $K^{O(K)}$ complete pure intervals;
\item
$V_{\tilde{t}_2}(a,(l,p,r),q,b) = V_{\tilde{t}_1}(a,(l,p,r),q,b)$,
for every $(a,(l,p,r),q,b) \in \Sigma\times\{\top,\bot,*\}^3\times Q\times \Gamma$;
\item
$\tilde{t}_2$ is an extended tree of its $\Sigma$ projection w.r.t. $\S$.
\end{itemize}
This step is adapted from~\cite[Proposition~3.14]{BMSS09}.
Again when the cut-and-paste is performed
the data values in the sibling groups remain untouched.

\item
Convert $\tilde{t}_2$ to another tree $\tilde{t}_3$ such that
\begin{itemize}
\item
for every data value $d \in V_{\tilde{t}_3}$
there are at most $O(K)$ $d$-zones
containing a path with more than $O(K)$ nodes;
\item
$V_{\tilde{t}_3}(a,(l,p,r),q,b) = V_{\tilde{t}_2}(a,(l,p,r),q,b)$,
for every $(a,(l,p,r),q,b) \in \Sigma\times\{\top,\bot,*\}^3\times Q\times \Gamma$;
\item
$\tilde{t}_3$ is an extended tree of its $\Sigma$ projection w.r.t. $\S$.
\end{itemize}
This step is adapted from~\cite[Proposition~3.17]{BMSS09}.
Again when the cut-and-paste is performed
the data values in the zones remain untouched.

\item
Convert $\tilde{t}_3$ to another tree $\tilde{t}_4$ such that
\begin{itemize}
\item
there are at most $K^{O(K^2)}$ complete pure intervals with more than $O(K^2)$ nodes;
\item
$V_{\tilde{t}_3}(a,(l,p,r),q,b) = V_{\tilde{t}_4}(a,(l,p,r),q,b)$,
for every $(a,(l,p,r),q,b) \in \Sigma\times\{\top,\bot,*\}^3\times Q\times \Gamma$;
\item
$\tilde{t}_4$ is an extended tree of its $\Sigma$ projection w.r.t. $\S$.
\end{itemize}
This step is adapted from~\cite[Proposition~3.20]{BMSS09}.
Here actually when the cut-and-paste is performed,
the data values in some zones have to be changed.
However, those changes are only applied to the {\em safe} zones,
where a zone is safe if for every node in it
there is another node outside the zone with the same label 
(from $\Sigma\times\{\top,\bot,*\}\times Q\times \Gamma$)
and the same data value. (See~\cite[page 23, last paragraph]{BMSS09}.) 
More specifically, these changes are done
by applying~\cite[Lemma 3.19]{BMSS09} on the safe zones.
That it is applied only on safe zones is important
so that after changing the data values,
constraints such as 
$\forall x \exists y (a(x) \to x\sim y \wedge b(y))$
are still satisfied.

\item
Convert $\tilde{t}_4$ to another tree $\tilde{t}_5$ such that
\begin{itemize}
\item
there are at most $K^{O(K^2)}$ sibling groups containing more than $K^{O(K)}$ complete pure intervals;
\item
$V_{\tilde{t}_4}(a,(l,p,r),q,b) = V_{\tilde{t}_5}(a,(l,p,r),q,b)$,
for every $(a,(l,p,r),q,b) \in \Sigma\times\{\top,\bot,*\}^3\times Q\times \Gamma$;
\item
$\tilde{t}_5$ is an extended tree of its $\Sigma$ projection w.r.t. $\S$.
\end{itemize}
This step is adapted from~\cite[Proposition~3.21]{BMSS09}.
Here there are also changes of data values when performing cut-and-paste.
However, as in the previous step, they are only applied to the {\em safe} zones.
These changes are also done
by applying~\cite[Lemma 3.19]{BMSS09} on the safe zones.

\item
Convert $\tilde{t}_5$ to another tree $\tilde{t}_6$ such that
\begin{itemize}
\item
there are at most $K^{O(K^2)}$ zones containing paths with more than $O(K^2)$ nodes;
\item
$V_{\tilde{t}_5}(a,(l,p,r),q,b) = V_{\tilde{t}_6}(a,(l,p,r),q,b)$,
for every $(a,(l,p,r),q,b) \in \Sigma\times\{\top,\bot,*\}^3\times Q\times \Gamma$;
\item
$\tilde{t}_6$ is an extended tree of its $\Sigma$ projection w.r.t. $\S$.
\end{itemize}
This step is adapted from~\cite[Proposition~3.25]{BMSS09}.
Here there are also changes of data values when performing cut-and-paste.
However, as in the previous step, they are only applied to the {\em safe} zones.
More specifically, these changes are done by applying~\cite[Lemma 3.24]{BMSS09} on the safe zones.
\end{enumerate}
The extended tree $\tilde{t}_6$ is the desired extended tree.
It is a rather straightforward computation that  
there are at most
$K^{O(K^2)}$ zones in $\tilde{t}_6$ with outdegree $\geq K ^{(K^3)}$.
\end{proof}

To describe the decision procedure for Theorem~\ref{t: decidability},
we need a few more additional terminologies.
For a data tree $t$ over the alphabet $\Gamma$,
and $S \subseteq \Gamma$,
an $S$-zone is a zone in which the labels of the nodes are precisely $S$.
We write $V^{zone}_t(S)$ to denote the set of data values
found in $S$-zones in $t$.
For $P \subseteq 2^{\Gamma}$,
$$
[P]^{zone}_t = \bigcap_{S \in P} V^{zone}_t(S) \cap \bigcap_{R \notin P} \overline{V^{zone}_t(R)}
$$
Suppose $d_1<\cdots < d_m$ are all the data values in $t$.
The {\em zonal string representation} of the data values in $t$, denoted by $\V^{zone}_{\Gamma}(t)$, 
is the string $P_1\cdots P_m$ over the alphabet $2^{2^{\Gamma}}$ 
such that for each $i\in\{1,\ldots,m\}$,
$d_i \in [P_i]^{zone}_t$.

A {\em zonal} ODTA is $\S' = \langle \T,\M',\Gamma_0\rangle$,
where $\T$ and $\Gamma_0$ are as in the definition of ODTA,
and $\M'$ is a finite state automaton over the alphabet $2^{2^{\Gamma}}$.
A data tree $t$ is accepted by the zonal ODTA $\S'$,
if the following holds.
\begin{itemize}
\item
$\sfProfile(t)$ is accepted by $\T$, 
yielding an output tree $t'$ over the alphabet $\Gamma$.
\item
The string $\V^{zone}_{\Gamma}(t')$ is accepted by $\M'$.
\item
For each $a\in \Gamma_0$,
all the data values found in the $a$-nodes in $t'$ are different.
\end{itemize}

\begin{proposition}
\label{p: convert to zonal}
For every ODTA $\S$,
one can construct in $\exptime$ its equivalent zonal ODTA.
\end{proposition}
\begin{proof}
Let $\S = \langle \T,\M,\Gamma_0\rangle$
and $\M = \langle Q, q_0, \delta, F\rangle$.
Its equivalent zonal ODTA is defined as 
$\S' = \langle \T,\M',\Gamma_0\rangle$,
where $\M' = \langle Q, q_0, \delta', F\rangle$ and
$\delta' = \{(q,P,q') \in Q \times 2^{2^{\Gamma}} \times Q \mid 
\exists (q,S,q')\in \delta \ \mbox{such that} \ \bigcup_{R \in P} R = S\}$.
It is straightforward to show that $\L_{data}(\S') = \L_{data}(\S)$.

Note that the only difference between $\S$ and $\S'$
is the transitions $\delta$ and $\delta'$ in $\M$ and $\M'$, respectively.
The membership $(q,P,q') \in \delta'$ can be checked in polynomial time
in the size of $(q,P,q')$ and $\delta$.
Since there are exponentially many $(q,P,q')$,
the exponential time upper bound holds immediately.
This completes the proof of Proposition~\ref{p: convert to zonal}.
\end{proof}

Briefly our decision procedure for Theorem~\ref{t: decidability} works as follows.
Let $\S = \langle \T,\M,\Gamma_0\rangle$ be the given ODTA,
where $\Sigma$ is the input alphabet of $\T$, $\Gamma$ the output alphabet,
and $Q$ the set of states of $\T$.
Let $K = 27\cdot |\Sigma|\cdot |Q|\cdot |\Gamma|$.
The decision procedure constructs an APC $(\A,\xi)$ such that
$\S$ accepts an ordered-data tree $t$ in which there are at most $K^{O(K^2)}$ zones
with outdegree $\geq K^{(K^3)}$ if and only if
$(\A,\xi)$ accepts an extended tree of $t$ w.r.t.~$\S$.

Its precise description is given as follows.
\begin{enumerate}
\item
Compute $K = 27 \cdot|\Sigma|\cdot|Q|\cdot|\Gamma|$.
\item
Convert $\S$ into its zonal ODTA 
$\S' = \langle \T, \M',\Gamma_0\rangle$.
\item
Guess the following items.
\begin{enumerate}
\item
A set $\P \subseteq 2^{2^{\Gamma}}$.
\item
For each $P \in \P$, guess an integer 
$M_P \leq 2\cdot K^{K^3}\cdot 2^K + 2\cdot K^{K^3}+1$
and a set of $M_P$ constants $\C_P =\{c_1,\ldots,c_{M_P}\}$.\footnote{The 
purpose of the number $2\cdot K^{K^3}\cdot 2^{K}+2\cdot K^{K^3}$
is the application of Lemma~\ref{l: canonical lemma} later on,
where we consider the graph where the nodes are the zones.
Each zone is labeled with a symbol from $2^{\Sigma\times\{\top,\bot,*\}^3\times Q\times \Gamma}$,
which is of size $2^{K}$.
If a zone has outdegree $\leq K ^{(K^3)}$,
then it has only at most $K ^{(K^3)}$ nodes,
which means that its degree (the sum of indegree and outdegree) is bounded by $2\cdot K^{K^3}$.
Now $\P$ is intended to contain all those $P$'s in which
$|[P]^{zone}_t| \leq 2\cdot K^{K^3}\cdot 2^K + 2\cdot K^{K^3}+1$ so that
we can ``guess'' some constants as elements of $[P]^{zone}_t$ and make sure
by automaton that the same constant is not ``assigned'' to adjacent zones.
For $P$ not in $\P$, we can apply Lemma~\ref{l: canonical lemma}
to make sure the same data value from $[P]^{zone}_t$ is not assigned to adjacent zones.}
\item
Two integers $N,N'$ such that $N' \leq N \leq K^{O(K^2)}$
and a set of $N'$ constants $\D = \{d_1,\ldots,d_{N'}\}$.
\\
The intuitive meaning of $N'$ and $N$ are the number
of zones with outdegree $\geq K^{(K^3)}$ and the number of data values
found in them, respectively.
We also remark that the constants in $\D$ may overlap with the constants in some $\C_P$.
\item
For each $d\in \D$, guess a set $P_d \subseteq 2^{\Gamma}$.
\end{enumerate}
\item
Construct the following automaton $\A$ over the alphabet 
$\Sigma\times\{\top,\bot,*\}^3\times Q \times \Gamma$.
\begin{enumerate}
\item
$\A$ accepts only the extended trees of $\L(\T)$
in which there are at most $N$ zones with outdegree $\geq K^{(K^3)}$.
\item
The automaton $\A$ can remember the constants in its states. 
\item
For every $P \in \P$, for every $c \in \C_P$,
the automaton $\A'$ ``assigns'' the constant $c$ in an $S$-zone, 
for every $S \in P$, but not in any $R$-zone, for every $R \notin P$.
\item
The automaton $\A$ ``assigns'' every zone with outdegree 
$\geq K^{(K^3)}$ with a constant from $\D$.
\item
For every $d\in \D$, for every $S\in P_d$,
the automaton $\A$ ``assigns'' the constant $d$ in an $S$-zone, 
for every $S \in P_d$, but in no $R$-zone, for every $R \notin P_d$.
\item
For each $a\in\Gamma_0$, there is at most one $a$-node in every zone,
and for every two zones that contains $a$-nodes,
if they are assigned with some constants from $\C_P$'s and $\D$,
then these constants must be different.
\item
For every two adjacent zones, if they are assigned with constants from $\C_P$'s and $\D$, 
then these constants must be different.
\end{enumerate}
The automaton $\A$ ``assigns'' a constant to a zone by remembering
the constant in the state when $\A$ is reading the zone.
\item
Let $P_1,\ldots,P_m$ be the enumeration of non-empty subsets of $2^{\Gamma}$.
\\
Applying Lemma~\ref{p: aut-to-presburger}, 
convert the automaton $\M'$ into its Presburger formula 
$\xi_{\sM'}(z_{P_1},\ldots,z_{P_m})$,
where the intended meaning of $z_{P_i}$'s is the number of appearances of
the label $P_i$.
\item
Let $\Gamma = \{a_1,\ldots,a_{\ell}\}$ and
$S_1,\ldots,S_k$ be the enumeration of non-empty subsets of $\Gamma$.
Define the formula 
$\xi(x_{a_1},\ldots,x_{a_{\ell}},x_{S_1},\ldots,x_{S_k}):$ 
\begin{eqnarray}
& \exists z_{P_1} \cdots \exists z_{P_m} & \xi_{\sM'}(z_{P_1},\ldots,z_{P_m}) 
\\
& & \ \wedge \ \bigwedge_{P_i \in \sP} z_{P_i} = M_{P_i}
\\ 
& & \ \wedge \ \bigwedge_{P_i \notin \sP} z_{P_i} \geq 
2\cdot K^{K^3}\cdot 2^K + 2\cdot K^{(K^3)} + 1
\\
& & \ \wedge \;\bigwedge_{S \subseteq \Gamma} 
\Big( x_S \ \geq \ \sum_{P_i \ni S \ \textrm{and} \ P_i \notin \sP} z_{P_i}\Big)
\\
& & \ \wedge \; \bigwedge_{a \in \Gamma_0} 
\Big( x_a \ = \  \sum_{\scriptsize\begin{array}{c}
 \textrm{there exists} \ S \ \textrm{such that}
 \\
 a \in S \ \textrm{and} \ S \in P_i\ \textrm{and} \ P_i \notin \P \end{array}} z_{P_i} \Big)
\\
& & \ \wedge \; 
\bigwedge_{\scriptsize P_i \notin \P} 
z_{P_i} \geq  |\{d \in \D \mid P_{d}=P_i\}|
\end{eqnarray}
The meaning of $x_a$ is the number of $a$-nodes occurring in the zone
not assigned with any constants from $\C_P$'s and $\D$;
and $x_S$ is the number $S$-zones not assigned
with any constants from $\C_P$'s and $\D$.
The intuition behind items~(2)--(6) is rather clear.
The intuition behind item~(7) is as follows.
Recall that in Step~(3), for each $d\in \D$,
we guess a set $P_d$.
The meaning is that $d \in [P_d]^{zone}_t$ for some $t\in \L^{data}(\S)$.
So for every $P_i \notin \P$,
the number of $d$ such that $P_d = P_i$ should not exceed $z_{P_i}$.
This is precisely what is stated in item~(7).
\item
Test the non-emptiness of the APC $(\A,\xi)$.
\end{enumerate}

Before we proceed to prove its correctness,
we first present the analysis of its complexity.
\begin{itemize}\itemsep=0pt
\item
Step~(1) is trivial and
Step~(2) takes exponential time.
\item
Step~(3) takes non-deterministic exponential time in the size of $\S$.
The analysis is as follows.
Step~(3.a) takes non-deterministic exponential time in the size of $2^{\Gamma}$,
which is bounded by the size of $\M$ in $\S$.
(Recall that the alphabet in $\M$ is $2^{\Gamma}$.)
Step~(3.b) can guess up exponentially many constant in each $\C_P$,
and there are exponentially many different $\C_P$,
hence it takes double exponential time in the size of $2^{\Gamma}$.
Steps~(3.c) and~(3.d) take non-deterministic exponential time.
\item
Step~(4) takes deterministic triple exponential time
and can produce the automaton $\A$ of size up to triple exponential.
The analysis is as follows.
The automaton $\A$ has to remember in its states the outdegree of each zone up to $K^{(K^3)}$
and the number of zones with out degree $\geq K^{(K^3)}$.
This induces an exponential blow-up in the size of $\T$.

The number of constants in guessed in Step~(3) is double exponential in the size of $\T$.
Then $\A$ has to remember in its states
which constant is assigned to which zone (of outdegree $\geq K^{(K^3)}$),
which induces another exponential blow-up.
Altogether the size of $\A$ can be triple exponential in the size of $\T$.

\item
By Proposition~\ref{p: aut-to-presburger},
Step~(5) takes polynomial time in the size $\M'$,
which is of size exponential in the size of the original $\M$.
 
\item
The length of the formula in step~(6) is double exponential in the size of $\S$,
since the number of constants in $\D$ can be double exponential in the size of $2^{\Gamma}$,
and hence $\S$.

\item
Step~(7) takes non-deterministic polynomial time in the size of $(\A,\xi)$,
and hence non-deterministic triple exponential time in the size of the input $\S$.
\end{itemize}

% We should remark here that
% there is a triple exponential blow-up in the size of $\A$,
% while double exponential blow-up in the size of $\xi$.
% Since the non-emptiness of APC is in $\np$,
% this yields a 3-$\nexp$ upper bound for our decision procedure.

The following claim immediately implies
the correctness of our algorithm.
\begin{claim}
\label{cl: correctness}
\begin{enumerate}
\item
For every ordered-data tree $t \in \L_{data}(\S)$,
in which there are at most $K^{O(K^2)}$ zones with outdegree $\geq K^{(K^3)}$,
there exists an extended tree of $t$ which is accepted by the APC $(\A,\xi)$.
\item
For every $t' \in \L(\A,\xi)$,
there exists an ordered-data tree $t \in \L_{data}(\S)$
such that $t'$ is an extended tree of $t$ w.r.t. $\S$.
\end{enumerate}
\end{claim}
\begin{proof}
We prove~(1) first.
Let $t\in \L_{data}(\S)$ be an ordered-data tree
in which there are at most $K^{O(K^2)}$ zones
with outdegree $\geq K^{(K^3)}$.
Let $t_0$ be the output of $\T$ on $t$
so that $\V^{zone}(t_0)$ is accepted by $\M$
and all nodes in $t_0$ labelled with a symbol in $\Gamma_0$
have different data values.

We have the following items guessed in Step~3 in our algorithm above.
\begin{itemize}
\item
$\P = \{P \mid |[P]^{zone}_t| \leq 2\cdot K^{K^3}\cdot 2^{K} + 2\cdot K^{(K^3)} + 1\}$.
\item
For each $P \in P$, $\C_{P} = [P]^{zone}_{t_0}$,
and $M_P = |\C_P|$.
\item
$N$ be the number of zones in $t$ with outdegree $\geq K^{(O(K^2))}$
and $N'$ be the number of data values found in these zones.
\item
$\D = \{d \mid d \ \mbox{is found in a zone with outdegree} \ \geq K^{(K^3)}\}$,
\item
For each $d\in\D$, $P_d$ is the set such that $d \in [P_d]^{zone}_{t_0}$.
\end{itemize}
Now let $t'$ be an extended tree of $t$ with respect to $\S$,
and $\A$ and $\xi$ be the automaton and formula as constructed in Steps~4--6 above.
We are going to show that $t' \in \L(\A,\xi)$.
Obviously, $t' \in \L(\A)$.
To show that the formula $\xi$ is satisfied,
we take $\Parikh(\V^{zone}(t_0))$ as witness to $(z_{P_1},\ldots,z_{P_m})$.
Since $\V^{zone}(t_0) \in \L(\M')$,
by Proposition~\ref{p: aut-to-presburger},
the formula $\xi_{\sM'}(\Parikh(\V^{zone}(t_0)))$ holds.
It is straightforward from the definitions of the items 
$\P$, $M_P$'s, $N$, $N'$, $\D$ and $P_d$'s
that the formula $\xi$ in Step~6 is satisfied
with $x_a$'s and $x_S$'s interpreted as intended.

Now we prove~(2). The proof is more delicate than the proof of~(1).
Suppose $t' \in \L(\A',\xi)$.
We are going to construct an ordered-data tree $t$ from $t'$
such that $t'$ is an extended tree of $t$ w.r.t. $\S$.
Let $\P$, $M_P$'s, $\C_P$'s, $N$, $N'$, $\D$ and $P_d$'s
the items as guessed in Step~3 above and
\begin{itemize}
\item
for each $a_i \in \Gamma$, let $n_{a_i}$
be the number of $a_i$-nodes in $t'$ occurring in a zone
without any constants from $\C_P$'s and $\D$;
\item
for each $S_i \subseteq \Gamma$, let $n_{S_i}$
be the number of $S_i$-zones in $t'$ without any constants from $\C_P$'s and $\D$.
\end{itemize}
Suppose $(k_{P_1},\ldots,k_{P_m})$ be the witness to $z_{P_1},\ldots,z_{P_m}$
such that 
$$
\xi(n_{a_1},\ldots,n_{a_{\ell}},n_{S_1},\ldots,n_{S_l}) \qquad\mbox{holds}.
$$
By Proposition~\ref{p: aut-to-presburger},
this means that there exists a word $w \in \L(\M')$ such that
$\Parikh(w)=(k_{P_1},\ldots,k_{P_m})$.
For each $P_i$, we let 
$$
\N_{P_i} = \{j \mid \mbox{position} \ j \ \mbox{in} \ w \ \mbox{is labeled} \ P_i\}.
$$

We will assign a data value to each node in $t$ such that 
$$
[P_i]^{zone}_{t} = \N_{P_i},
$$
and $\V^{zone}(t) = w$.
The assignment is done according to three cases below.
\begin{description}
\item[Case 1]~
For the nodes that are assigned with some constants from $\C_{P_i}$'s.
\\
In this case $P_i \in \P$.
We define bijections $f_{P_i}:\C_{P_i} \mapsto \N_{P_i}$.
There is always a bijection from $\C_{P_i}$ to $\N_{P_i}$ since
they have the same cardinality $M_{P_i}$, due to the following condition in the formula $\xi$:
$$
\bigwedge_{P_i \in \sP} z_{P_i} = M_{P_i}.
$$
The data value assignment to nodes of this case can be done
by replacing every constant $c\in \C_{P_i}$ with $f_{P_i}(c)$.
\item[Case 2]~
For the nodes that are assigned some constants from $\D$.
\\
We define a 1-1 mapping $f : \D \mapsto \{1,\ldots,|w|\}$
such that $f(d) \in \N_{P_d}$,
where $P_d$ is the set guessed in Step~3.
Such 1-1 mapping exists because the following condition in the formula $\xi$: 
$$
\bigwedge_{\scriptsize P_i \notin \P} 
z_{P_i} \geq  |\{d' \in \D \mid P_{d'}=P_i\}|
$$
The data value assignment to nodes of this case can be done
by replacing every constant $d\in \D$ with $f(d)$.
\item[Case 3]~
For the nodes that are not assigned any constants from $\C_P$'s and $\D$.
\\
First we assign each of such zone in $t$ with a data value\footnote{A zone in $t$ can be recognised 
from the profile information in $t'$.} such that
for each $S\subseteq \Gamma$,
$$
V^{zone}_{t}(S) = \bigcup_{P_i \ni S \ \textrm{and} \ P_i\notin \sP} \N_{P_i}
$$
This step can be done as follows. The number of such $S$-zone in $t$ 
is greater than $\sum_{P_i \ni S \ \textrm{and} \ P_i\notin \sP} |\N_{P_i}|$,
due to the condition below in the formula $\xi$:
$$
x_S \ \geq \ \sum_{P_i \ni S \ \textrm{and} \ P_i \notin \sP} z_{P_i}.
$$
Thus, we can simply assign every $S$-zone with a data value from 
$\bigcup_{P_i \ni S \ \textrm{and} \ P_i\notin \sP} \N_{P_i}$,
and make sure every data value from $\bigcup_{P_i \ni S \ \textrm{and} \ P_i\notin \sP} \N_{P_i}$
appears in some $S$-zone.
\\
However, by assigning data values like that, some adjacent zones may
get the same data values.
Here we apply Lemma~\ref{l: canonical lemma}.
Since for each $P_i\notin \P$, 
$|\N_{P_i}| \geq 2\cdot K^{K^3}\cdot 2^K + 2\cdot K^{(K^3)} + 1$,
by the condition below in the formula $\xi$
$$
\bigwedge_{P_i \notin \sP} z_{P_i} 
\geq 2\cdot K^{K^3}\cdot 2^K + 2\cdot K^{(K^3)} + 1,
$$
the cardinality
$$
\Big|\bigcup_{P_i \ni S \ \textrm{and} \ P_i\notin \sP} \N_{P_i}\Big| = 
\sum_{P_i \ni S \ \textrm{and} \ P_i\notin \sP} |\N_{P_i}| \geq 
2\cdot K^{K^3}\cdot 2^K + 2\cdot K^{(K^3)} + 1.
$$
The outdegree of such zone is $\leq K^{(K^3)}$,
hence, the number of nodes in the zone is also $\leq K^{(K^3)}$.
Since each node can have indegree at most $1$,
the degree of each of such zone is $\leq 2\cdot K^{(K^3)}$.
By applying Lemma~\ref{l: canonical lemma}, where $\deg(G)= 2\cdot K^{(K^3)}$,
we can reassign the data value in such zone so that
each adjacent zone get different data value.
\end{description}
This completes the proof of our Claim.
\end{proof}

\section{Weak ODTA}
\label{s: weak S-automata}

A weak ODTA over $\Sigma$ is a triplet $\S = \langle \T, \M, \Gamma_0\rangle$
where $\T$ is a letter-to-letter transducer from $\Sigma$
to the output alphabet $\Gamma$,
and $\M$ is a finite state automaton over $2^{\Gamma}$ and 
$\Gamma_0\subseteq \Gamma$.
An ordered-data tree $t$ is accepted by $\S$,
denoted by $t \in \L_{data}(\S)$, if there exists an ordered-data tree $t'$ over $\Gamma$
such that
\begin{itemize}\itemsep=0pt
\item
on input $\sfProj(t)$, the transducer $\T$ outputs $t'$;
\item
the automaton $\M$ accepts the string $\V_{\Gamma}(t')$; and
\item
for every $a\in \Gamma_0$,
all the $a$-nodes in $t'$ have different data values.
\end{itemize}
Note that the only difference between weak ODTA and ODTA
is the equality test on the data values in neighboring nodes.
Such difference is the cause of the triple exponential leap in complexity,
as stated in the following theorem.

\begin{theorem}
\label{t: decidability weak S-automata}
The non-emptiness problem for weak ODTA is in $\np$.
\end{theorem}
\begin{proof}
Let $\S = \langle \T,\M,\Gamma_0\rangle$ be a weak ODTA.
Let $\Sigma,Q,\Gamma$ be the input alphabet, set of states
and output alphabet of $\T$, respectively.

We need the following notation.
For a tree $t \in \L_{data}(\S)$,
its extended tree $\tilde{t}$ (with respect to the weak ODTA $\S$)
is a tree over the alphabet $\Sigma\times Q \times \Gamma$,
where
\begin{itemize}
\item
the projection of $\tilde{t}$ to $\Sigma$ is $t$;
\item
the projection of $\tilde{t}$ to $Q$ is an accepting run of $\T$ on $t$
such that its output is 
the projection of $\tilde{t}$ to $\Gamma$.
\end{itemize}

The decision procedure for Theorem~\ref{t: decidability weak S-automata} works as follows.
\begin{enumerate}
\item
Construct an automaton $\A$ over the alphabet $\Sigma\times Q \times \Gamma$
for the extended trees accepted by $\T$.
\item
Let $\P =\{S_1,\ldots,S_m\} \subseteq 2^{\Gamma}$ be the set of symbols used in $\M$.
\\
By applying Proposition~\ref{p: aut-to-presburger},
construct the Presburger formula $\xi_{\sM}(x_{S_1},\ldots,x_{S_m})$ for $\M$.
\item
Let $\Sigma\times Q \times \Gamma =\{(a_1,q_1,\alpha_1),\ldots,(a_k,q_n,\alpha_{\ell})\}$.
Let $\varphi(x_{(a_1,q_1,\alpha_1)},\ldots,x_{(a_k,q_n,\alpha_{\ell})})$ be the following formula:
\begin{eqnarray*}
& & \exists x_{\alpha_1} \cdots \exists x_{\alpha_{\ell}} \
\exists x_{S_1} \cdots \exists x_{S_m} 
\ \xi_{\sM}(x_{S_1},\ldots,x_{S_m})
\\
& & \quad\ \wedge \
\bigwedge_{\alpha_i \in \Gamma} 
\Big( x_{\alpha_i} = \sum_{a_j\in\Sigma,q_h\in Q} x_{(a_j,q_h,\alpha_i)}\Big)
\\
& & \quad \ \wedge 
\bigwedge_{\alpha_i \in \Gamma} 
\Big( x_{\alpha_i} \geq \sum_{\alpha_i \in S_j} x_{S_j} \Big)
\ \wedge \ 
\bigwedge_{\alpha_i \in \Gamma_0} 
\Big( x_{\alpha_i} = \sum_{\alpha_i \in S_j} x_{S_j} \Big).
\end{eqnarray*}
\item
Test the non-emptiness of APC $(\A,\varphi(x_{(a_1,q_1,\alpha_1)},\ldots,x_{(a_k,q_n,\alpha_{\ell})}))$.
\end{enumerate}
That this procedure works in $\np$ follows directly from the fact
that the non-emptiness problem of APC is in $\np$.

We now show the correctness of our algorithm by showing that
$\L_{data}(\S) \neq \emptyset$ if and only if
$\L(\A,\varphi) \neq \emptyset$. (For the sake of presentation,
we write $\varphi$ without its free variables.)
We start with the ``only if'' part.
Suppose that $t \in \L_{data}(\S)$.
We claim that the extended tree $\tilde{t}$ of $t$ is accepted by $(\A,\varphi)$.
Obviously, $\tilde{t}\in \L(\A)$.
To show that $\varphi(\Parikh(\tilde{t}))$ holds,
let $t'$ be the data tree obtained by projecting $\tilde{t}$ to $\Gamma$
and the data value in each node comes from the same node in $t$.
That is, $t'$ is an output of $\T$ on $t$.
We will show that $\varphi(\Parikh(\tilde{t}))$ holds.
\begin{itemize}
\item
As witness to $x_{S_1},\ldots,x_{S_m}$,
we take $\Parikh(\V(t'))$.
Since $\V(t') \in \L(\M)$, by Proposition~\ref{p: aut-to-presburger},
$\xi_{\sM}(\Parikh(\V(t')))$ holds.
\item
As witness to $x_{\alpha_1},\ldots,x_{\alpha_{\ell}}$, we take $\Parikh(t')$.
Now for each $\alpha_i \in \Gamma$, the constraint
$x_{\alpha_i} \geq \sum_{\alpha_i \in S_j} x_{S_j}$ holds
since the number of data values in the $\alpha_i$-nodes
cannot exceed the the number of $\alpha_i$-nodes itself.
The constraint $x_{\alpha_i} = \sum_{\alpha_i \in S_j} x_{S_j}$,
for each $\alpha_i \in \Gamma_0$, since
the data values found in $\alpha_i$-nodes are all different.
\end{itemize}
Thus, $\varphi(\Parikh(\tilde{t}))$ holds, and this concludes
our proof of the ``only if'' part.

Now we prove the ``if'' part.
Suppose that $\tilde{t}\in \L(\A,\varphi)$.
So $\tilde{t} \in \L(\A)$.
Let $t$ and $t'$ be the $\Sigma$- and $\Gamma$-projection of $\tilde{t}$, respectively.
By the definition of $\A$,
$t'$ is an output of $\T$ on $t$.
Now since $\varphi(\Parikh(\tilde{t}))$ holds,
in particular there exists a witness $\bar{M}=(M_1,\ldots,M_m)$ to $x_{S_1},\ldots,x_{S_m}$
such that $\xi_{\sM}(\bar{M})$ holds,
by Proposition~\ref{p: aut-to-presburger},
there exists a word $w \in \L(\M)$ over the alphabet $2^{\Gamma}$
such that $\Parikh(w)=\bar{M}$.

We are going to assign data values to the nodes of $t'$ (thus, also to those of $t$)
such that $t \in \L_{data}(\S)$.
The assignment is done as follows.
For each $S\subseteq \Gamma$,
let $V_w(S)$ be the set of positions of $w$ labeled with $S$.
Now for each $\alpha \in \Gamma$,
we assign the $\alpha$-nodes in $t'$ with the data values from $\bigcup_{\alpha\in S} V_w(S)$
such that $V_{t'}(\alpha) = \bigcup_{\alpha\in S} V_w(S)$.
This is possible due to the constraint
$x_{\alpha} \geq \sum_{\alpha\in S} x_S$.

With such assignment, we get $\V(t') = w$.
Thus, $\V(t') \in \L(\M)$.
Moreover, for every $\alpha\in\Gamma_0$,
all the data values in $\alpha$-nodes are different,
which follows from the constraint $x_{\alpha} = \sum_{\alpha\in S} x_S$.
Therefore, the resulting ordered-data tree $t \in \L_{data}(\S)$.
This concludes our proof.
\end{proof}

Next, we give the logical characterisation of weak ODTA.

\begin{theorem}
\label{t: logic weak S-automata}
A language $\L$ is accepted by a weak ODTA if and only if
$\L$ is expressible with a formula of the form:
$\exists X_1 \cdots \exists X_m \ \varphi \wedge \psi$,
where $\varphi$ is a formula from $\FO^2(\eda,\era)$,
and $\psi$ is a formula from $\FO(\sim,\prec,\dvSucc)$,
extended with the unary predicates $X_1,\ldots,X_m$.
\end{theorem}

The proof of Theorem~\ref{t: logic weak S-automata} 
is the same as the proof of Theorem~\ref{t: logic S-automata}.
The difference is that to simulate the $\FO^2(\eda,\era)$ formula $\varphi$,
the profile information is not necessary.
The complexity of the translation is still the same as in Theorem~\ref{t: logic S-automata}.

\subsection{Extending weak ODTA with Presburger constraints}
\label{ss: extending}

Like in the case of APC,
we can extend weak ODTA with Presburger constraints
without increasing the complexity of its non-emptiness problem.
Let $\S = \langle \T, \M, \Gamma_0 \rangle$
be a weak ODTA, where
$\Sigma$ and $\Gamma$ are the input and output alphabets of $\T$, respectively.
Let $\Gamma = \{\alpha_1,\ldots,\alpha_{\ell}\}$.

A weak ODTA $\S = \langle \T,\M,\Gamma_0\rangle$ extended with Presburger constraint
is a tuple $\langle \S,\xi\rangle$,
where $\xi(x_1,\ldots,x_{\ell},y_1,\ldots,y_{2^{\ell}-1})$ 
is an existential Presburger formula
with the free variables $x_1,\ldots,x_{\ell},y_1,\ldots,y_{2^{\ell}-1}$.
An ordered-data tree $t$ is accepted by $\langle \S,\xi\rangle$,
if there exists an output $t'$ of $\T$ on $t$,
the automaton $\M$ accepts $\V_{\Gamma}(t')$,
for each $a\in \Gamma_0$, all $a$-nodes in $t'$
have different data values and
$\xi(\Parikh(t'),\Parikh(\V_{\Gamma}(t')))$ holds.
We write $\L_{data}(\S,\xi)$ to denote
the set of languages accepted by $\langle \S,\xi\rangle$.

We claim that the non-emptiness problem
of weak ODTA extended with Presburger constraint
is still decidable in $\np$.
The reason is as follows.
The non-emptiness of a weak ODTA $\S$ is checked by
converting $\S$ into an APC $(\A,\varphi)$,
where $\varphi$ expresses linear constraints on
the number of nodes labeled with symbols from $\Sigma$ and $\Gamma$
as well as those labeled with $Q$ in the accepting run.
The formula $\xi$ can be appropriately ``inserted'' into $\varphi$,
and hence, the non-emptiness of $(\S,\xi)$
is reducible to non-emptiness of APC,
which is in $\np$.

\subsection{Comparison with other known decidable formalisms}
\label{ss: comparison}

We are going to compare 
the expressiveness of weak ODTA with other known
models with decidable non-emptiness.

\subsubsection{DTD with integrity constraints}

An XML document is typically viewed as a data tree.
The most common XML formalism is Document Type Definition (DTD).
In short, a DTD is a context free grammar and a tree $t$ conforms to a DTD $D$,
if it is a derivation tree of a word accepted by the context free grammar.

The most commonly used XML constraints are integrity constraints
which are of two types.
\begin{itemize}
\item
The {\em key constraint} $key(a)$ are the following constraint:
$$
\forall x \forall y (a(x) \wedge a(y) \wedge x\sim y \to x=y).
$$
\item
The {\em inclusion constraint} $V(a)\subseteq V(b)$ are the following constraint:
$$
\forall x \exists y (a(x) \to  b(y) \wedge x\sim y ).
$$
\end{itemize}
The satisfiability problem of a given DTD $D$ and a collection $\C$ of integrity constraints
asks whether there exists an ordered-data tree $t$ that conforms to the DTD
that satisfies all the constraints in $\C$. 
In~\cite{FL-jacm} it is shown that this problem is $\np$-complete.

\begin{theorem}
\label{t: precise dtd}
Given a DTD $D$ and a collection $\C$ of integrity constraints,
one can construct a weak ODTA $\S$
such that $\L_{data}(\S)$ is precisely the set of ordered-data trees
that conforms to $D$ and satisfies all constraints in $\C$.
\end{theorem}
\begin{proof}
Let $\Sigma$ be the alphabet of the given DTD $D$.
Consider the following weak ODTA $\S = \langle \T,\M,\Sigma_0 \rangle$.
\begin{itemize}\itemsep=0pt
\item
$\T$ is an identity transducer that checks whether the input tree conforms to DTD $D$.
\item
$\M$ is an automaton that accepts $\P^{\ast}$, where
$\P = 2^{\Sigma} - \{ S \mid a \in S \ \mbox{and} \ b \notin S
\ \mbox{for some} \ V(a)\subseteq V(b) \in \C \}$.
\item
$\Sigma_0 = \{a \mid key(a)\in \C\}$.
\end{itemize}
That $\S$ is the desired ODTA follows immediately from the fact
that for every ordered-data tree $t$,
$V_t(a)\subseteq V_t(b)$ if and only if
$[S]_t = \emptyset$ for all $S$ where $a\in S$, but $b\notin S$.
\end{proof}

The size of the automaton $\M$,
hence the size of $\S$,
produced by our construction in Theorem~\ref{t: precise dtd}
is of exponential size.
This blow-up is tight, as the following example shows.
Consider the case where $\C$ does not contain inclusion constraints.
That is, $\C$ contains only key constraints.
Then any equivalent ODTA $\S = \langle \T,\M,\Sigma_0\rangle$
will have $\L(\M) = (2^{\Sigma}-\{\emptyset\})^{\ast}$.
Thus, we have exponential blow-up in the size of $\M$.
Nevertheless, if we are concerned only with
satisfiability, then we can lower the complexity to $\np$
as stated in the following theorem.

\begin{theorem}
\label{t: dtd}
Given a DTD $D$ and a collection $\C$ of integrity constraints,
one can construct a weak ODTA $\S$ in non-deterministic polynomial time
such that $\L_{data}(\S)\neq \emptyset$ if and only if
there exists an ordered-data tree $t$ that 
conforms to $D$ and satisfies all the constraints in $\C$.
\end{theorem}
\begin{proof}
Let $\Sigma$ be the alphabet of the DTD $D$.
We non-deterministically construct a weak ODTA 
$\S = \langle \T,\M,\Sigma_0\rangle$ as follows.
\begin{itemize}
\item
$\T$ is an identity transducer that checks whether the input tree conforms to DTD $D$.
\item
Guess a sequence $(H_1,\ldots,H_k)$ of {\em some} subsets of $\Sigma$ 
such that
\begin{itemize}\itemsep=0pt
\item
$\Sigma$ is partitioned into $H_1\cup\cdots\cup H_k$;
\item
for every two different symbols $a,b \in \Sigma$,
$a,b$ are in the same set $H_i$ if and only if 
both $V(a)\subseteq V(b)$ and $V(b)\subseteq V(a)$ are in $\C$;
\item
if $V(a)\subseteq V(b) \in \C$ and
$V(b)\subseteq V(a) \not\in \C$,
then $a \in H_i$ and $b \in H_j$ and $i \leq j$.
\end{itemize}
Intuitively, the sequence $(H_1,\ldots,H_k)$
tells us the ordering of the elements in $\Sigma$
that respect the inclusion constraints in $\C$,
where if both $V(a)\subseteq V(b)$ and $V(b)\subseteq V(a)$ are in $\C$,
then $a$ and $b$ are tie and they must be in the same set $H_i$.
\item
Let $S_1,\ldots,S_k\subseteq \Sigma$ be such that $S_i = \Sigma - (H_1\cup \cdots \cup H_{i-1})$,
where $S_1=\Sigma$ and $S_k=H_k$.
\item
$\M$ is a non-deterministic automaton over the alphabet $\{S_1,\ldots,S_k\}$, where
the set of states is $\{q_1,\ldots,q_k\}$,
all $q_1,\ldots,q_k$ are the initial states and the final states,
and the transitions are: $(q_i,S_j,q_j)$ for every $1\leq i \leq j \leq k$.
\item
$\Sigma_0 = \{a \mid key(a) \in \C\}$.
\end{itemize}
We claim that
$\L_{data}(\S)\neq \emptyset$ if and only if
there exists an ordered-data tree $t$ that conforms to $D$ and 
satisfies all the constraints in $\C$.

We start with the ``if'' direction.
Suppose $t$ conforms to the DTD $D$ and satisfies all the constraints in $\C$.
For each $a\in \Sigma$,
let $N_a$ be the number of data values found in the $a$-nodes in $t$.
Let $(H_1,\ldots,H_k)$ be a sequence of {\em some} subsets of $\Sigma$ 
such that
\begin{itemize}\itemsep=0pt
\item
$\Sigma$ is partitioned into $H_1\cup\cdots\cup H_k$;
\item
for every two different symbols $a,b \in \Sigma$,
$a,b$ are in the same set $H_i$ if and only if 
$N_a = N_b$;
\item
$a \in H_i$ and $b \in H_j$ and $i \leq j$
if and only if
$N_a \leq N_b$.
\end{itemize}

Consider the following ordered-data tree $t'$ over $\Sigma$, where
$t'$ is obtained from $t$ by reassigning the data values on the nodes in $t$ 
as follows.
For each $a\in \Sigma$,
we assign the set of integers $\{d \mid 1 \leq d \leq N_a\}$
as the data values of $a$-nodes in $t'$.
Such assignment is possible since $N_a$ is no more than
the number of $a$-nodes in $t'$.
With such assignment $t'$ still obeys the constraints in $\C$,
as shown below.
\begin{itemize}
\item
If $key(a)\in \C$, then
$N_a$ is precisely the number of $a$-nodes in $t$, thus, also in $t'$.
Thus, with the data values $\{1,\ldots,N_a\}$,
the data values on the $a$-nodes in $t'$ are all different.
\item
If $V(a)\subseteq V(a') \in \C$,
then obviously, $N_a \leq N_{a'}$.
Thus, $t'$ still satisfies the constraint $V(a)\subseteq V(a')$,
since the data values in $a$-nodes in $t'$ are $\{1,2,\ldots,N_a\}$,
while those in $a'$-nodes are $\{1,2,\ldots,N_{a'}\}$.
\end{itemize}
Now the string $\V(t')$ is of the form $R_1\cdots R_m$, where $m=\max_{a\in \Sigma}(N_a)$
where $R_1 \supseteq R_2 \supseteq \cdots \supseteq R_m$,
and if $R_i \neq R_{i+1}$,
then $R_{i+1}-R_i = H_j$ for some $H_j$ in the sequence $(H_1,\ldots,H_k)$.
By the definition of $\M$,
$\V(t')$ is accepted by $\M$.
That $t$ is accepted by $\T$ is trivial and so is the fact that 
all the data values found in $a$-nodes in $t'$ for each $a \in \Sigma_0$.
Thus, $t' \in \L_{data}(\S)$.

For the ``only if'' direction,
it is sufficient to observe that
for every sequence $(H_1,\ldots,H_k)$ that ``respects''
the inclusion constraints in $\C$ as explained above,
if $\V(t) \in \L(\M)$,
then $t$ satisfies all the inclusion constraints in $\C$.
This completes our proof.
\end{proof}

\subsubsection{Set and linear constraints for data trees}

In the paper~\cite{edt} the {\em set and linear constraints} 
are introduced for data trees.
As argued there, those constraints, together with automata,
 are able to capture many interesting properties commonly used in XML practice.
We review those constraints and show how they can be captured by 
weak ODTA extended with Presburger constraints.

{\em Data-terms} (or just terms) are given by the grammar
$$
\tau := V(a)\ |\ \tau\cup\tau \ |\  \tau \cap \tau \ |\ \overline{\tau}\quad
\quad \mbox{for} \ a\in\Sigma.
$$
The semantics of $\tau$ is defined with respect to a data tree $t$:
$$
\begin{array}{lll}
\semt {{V(a)}} = V_t(a) & \quad & \semt{{\overline{\tau}}}=V_t-\semt{\tau}
\\
\semt {{\tau_1\cap\tau_2}}=\semt{{\tau_1}}\cap\semt{{\tau_2}}
& \quad &
\semt{{\tau_1\cup\tau_2}}=\semt{{\tau_1}}\cup\semt{{\tau_2}}
\end{array}
$$
Recall that $V_t=\bigcup_{a\in\Sigma}V_t(a)$ -- the set of data values
found in the data tree $t$.

A {\em set constraint} is either $\tau = \emptyset$ or
$\tau\neq\emptyset$, where $\tau$ is a term. 
A data tree $t$ satisfies $\tau=\emptyset$,
 written as $t\models\tau=\emptyset$,
if and only if $\semt {\tau}=\emptyset$ (and likewise for $\tau\neq\emptyset$). 

A {\em linear constraint} $\xi$ over the alphabet $\Sigma$
is a linear constraint on the variables $x_a$, for each $a\in \Sigma$
and $z_S$, for each $S \subseteq \Sigma$.
A data tree $t$ satisfies $\xi$, if $\xi$ holds
by interpreting $x_a$ as the number of $a$-nodes in $t$,
and $z_S$ the cardinality $|[S]_t|$.

\begin{theorem}
\label{t: constraints into S-automata}
Given a tree automaton $\A$
and a set $\C$ of set and linear constraints,
there exists a weak ODTA $\langle\S,\varphi\rangle$ extended with Presburger constraints
such that $\L_{data}(\S,\varphi)$ is precisely
the set of ordered-data trees accepted by $\A$
that satisfies all the constraints in $\C$.
Moreover, the construction of $\langle S,\varphi\rangle$
takes exponential time in the size of $\A$ and $\C$.
\end{theorem}
\begin{proof}
The proof is simply a restatement of the proof in~\cite{edt}
into a language of weak ODTA.
We need the following notation.
For a data term $\tau$, 
we define a family $\bbS(\tau)$ of subsets of ${\Sigma}$ as follows.
\begin{itemize}
\item
If $\tau = V(a)$, then
$\bbS(\tau) = \{S \mid a \in S \ \mbox{and} \ S \subseteq \Sigma\}$.
\item
If $\tau = \overline{\tau}_1$,
then $\bbS(\tau) = 2^{\Sigma} - \bbS(\tau_1)$.
\item
If $\tau = \tau_1 \star \tau_2$,
then $\bbS(\tau) = \bbS(\tau_1) \star \bbS(\tau_2)$, where $\star$ is
$\cap$ or $\cup$.
\end{itemize}
It follows that for every data tree $t$, we have
$\semt{{\tau}} = \bigcup_{S\in \sbbS(\tau)} [S]_t$.
Recall that the sets $[S]_t$'s are disjoint. 

The desired $\S = \langle \T,\M,\Sigma_0\rangle$ is defined as follows.
The transducer $\T$ is the identity transducer $\A$, and
$\Sigma_0 = \emptyset$.
The automaton $\M$ accepts a word $v \in (2^{\Sigma})^{\ast}$
if and only if
\begin{itemize}
\item[C1.]
for every set constraint $\tau = \emptyset$,
$v$ does not contain any symbol from $\bbS(\tau)$;
\item[C2.]
for every set constraint $\tau \neq \emptyset$,
$v$ contains at least one symbol from $\bbS(\tau)$.
\end{itemize}
The formula $\xi$ is the conjunction of all the linear constraints in $\C$.

That $\L_{data}(\S,\xi)$ is indeed precisely
the set of ordered-data trees accepted by $\A$
that satisfies all the constraints in $\C$
follows immediately from the definition of $\bbS$.
The exponential upper-bound occurs while constructing the automaton $\M$
which requires the enumeration of each element of $2^{\Sigma}$
and checking both conditions C1 and C2 are satisfied.
This completes the proof of Theorem~\ref{t: constraints into S-automata}.
\end{proof}

\subsubsection{$\FO^2(+1,\dvSucc)$ over text}
\label{ss: fo2 successor data value}

Here we focus our attention on ordered-data words,
which can be viewed as trees where 
each node has at most one child.
We write $w = {a_1 \choose d_1}\cdots {a_n \choose d_n}$
to denote ordered-data word
in which position $i$ has label $a_i$ and data value $d_i$.
It is called a {\em text}, if
all the data values are different and
the set of data values $\{d_1,\ldots,d_n\}$
is precisely $\{1,\ldots,n\}$.

It is shown in~\cite{fo2-amal} that
the satisfiability problem for $\FO^2(+1,\dvSucc)$
over text is decidable.\footnote{The definition of text in~\cite{fo2-amal}
is slightly different, but it is equivalent 
to our definition. However, it turns out that
the key lemma proved in~\cite{fo2-amal} has a gap
which is filled later on in~\cite{figueira-arxiv}.
The final result is still correct though.}
The following theorem shows that
this decidability can be obtained via weak ODTA.

\begin{theorem}
\label{t: text automata to S-automata}
For every formula $\varphi \in \FO^2(+1,\dvSucc)$,
one can construct effectively a weak ODTA $\S$ such that 
\begin{itemize}
\item
for every text $w$, if $w \in \L_{data}(\varphi)$,
then $w \in \L_{data}(\S)$;
\item
for every ordered-data word $w \in \L_{data}(\S)$,
there exists a text $w' \in \L_{data}(\varphi)$ such that
$\sfProj(w)=\sfProj(w')$.
\end{itemize}
The construction of $\S$ takes double exponential time in the size of $\varphi$.
\end{theorem}
\begin{proof}
In~\cite{fo2-amal}, 
the decidability is proved by constructing its so called text automata,
also defined in~\cite{fo2-amal}.
We review the precise definition here.
Let $w={a_1 \choose d_1} \cdots {a_n \choose d_n}$ be a text over the alphabet $\Sigma$.
Therefore, $\V(w) = S_1\cdots S_n$ is such that each $S_i$ is a singleton.

We define $msp(w)$, the marked string projection of $w$,
as the word $(a_0,b_0)\ldots (a_n,b_n)$, where $b_i \in \{-1,1,*\}$
and
$$
b_i =
\left\{
\begin{array}{ll}
-1 & \mbox{if} \ 1 \leq i < n \ \mbox{and} \ d_{i+1} + 1 = d_i
\\
1 & \mbox{if} \ 1 \leq i < n \ \mbox{and} \ d_{i} + 1 = d_{i+1}
\\
* & \mbox{otherwise}
\end{array}
\right.
$$
A text automaton over the alphabet $\Sigma$ is pair $(T_1,T_2)$,
where
\begin{itemize}
\item 
$T_1$ is a non-deterministic letter-to-letter 
word transducer with the input alphabet $\Sigma\times\{-1,1,*\}$
and the output alphabet $\Gamma$.
\item
$T_2$ is a non-deterministic finite state automaton over $\Gamma$.
\end{itemize}
A text $w = {a_1 \choose d_1}\cdots {a_n \choose d_n}$
is accepted by the text automaton $(T_1,T_2)$,
if
\begin{itemize}
\item
$msp(w)$ is accepted by $T_1$, yielding a string $\alpha_1 \cdots \alpha_n$;
\item
the string $\alpha_{i_0}\cdots \alpha_{i_n}$ is accepted by $T_2$,
where the indexes $i_1,\ldots,i_n$ are such that
$1=d_{i_1}< d_{i_2} < \cdots < d_{i_n} = n$.
\end{itemize}
It is shown in~\cite{fo2-amal} that
for every $\varphi \in \FO^{2}(+1,\dvSucc)$,
one can construct effectively a text automaton $\A$
such that for every text $w$, $w \in \L_{data}(\varphi)$
if and only if $w \in \L_{data}(\A)$.

Now we are going to show how to get the desired ODTA $\S = \langle \T,\M,\Gamma\rangle$.
Let $(T_1,T_2)$ be the text automaton as above.
On input ordered-data word $w = {a_1 \choose d_1}\cdots {a_n \choose d_n}$,
$\S$ performs the following.
\begin{itemize}
\item 
The automaton $\T$ simulates $T_1$, 
by guessing $msp(w)$ and outputs its $\Gamma$-projection,
while store its $\{-1,1,*\}$-projection in its states.
\item
The automaton $\M$ is simply $T_2$.
\end{itemize}
It is straightforward to see that
such $\S$ is the desired weak ODTA.
The analysis of the complexity is as follows.
The construction of the text automaton $(T_1,T_2)$ takes double exponential time
in the size of $\varphi$.
See~\cite[Lemmas~5 and~6]{fo2-amal}.
The construction of ODTA $\S$ takes polynomial time in the size of $(T_1,T_2)$.
Altogether, it takes double exponential time to construct $\S$ from the original formula $\varphi$.
This completes the proof of Theorem~\ref{t: text automata to S-automata}.
\end{proof}

\section{An Undecidable Extension}
\label{s: undecidable}

In this section we would like to remark on an undecidable extension of ODTA.
Recall the language in Example~\ref{eg: two a-nodes}.
It has already noted in the proof of Proposition~\ref{p: boolean closure}
that its complement is not accepted by any ODTA.
Formally, the complement of the language in Example~\ref{eg: two a-nodes}
can be expressed with formula of the form:
\begin{equation}
\label{eq: monotone}
\forall x \;\forall y \;
\bigvee_{a\in \Sigma_0} a(x)
\wedge
\bigvee_{a\in \Sigma_0} a(y)
\wedge
\eda^*(x,y)
\to
x \prec y,
\end{equation}
where $\Sigma_0 \subseteq \Sigma$ and $\eda^*$ denotes the transitive closure of $\eda$.
In the following we are going to show that
given an ODTA and a collection $\C$ of formulas of the form~(\ref{eq: monotone}),
it is undecidable to check whether 
there is an ordered-data tree $t \in \L_{data}(\S)$ such that 
$t \models \psi$, for all $\psi \in \C$.

The proof is simply an observation that the proof of~\cite[Proposition 29]{BDMSS11}
can be applied directly here.
In~\cite[Proposition 29]{BDMSS11} it is proved that
the satisfiability of $\FO^2(\eda,\eda^*,\sim,\prec)$ is undecidable.\footnote{Technically,
the undecidability in~\cite[Proposition 29]{BDMSS11}
is proved on data strings over the logic $\FO^2(+1,<,\sim,\prec)$,
which of course, is equivalent to $\FO^2(\eda,\eda^*,\sim,\prec)$.}
The reduction is from Post Correspondence Problem (PCP),
where given an instance of PCP,
one can effectively construct a formula of the form $\varphi\wedge\psi$,
where $\varphi \in \FO^2(\eda,\eda^*,\sim)$ and 
$\psi$ is a formula of the form~(\ref{eq: monotone}).
Since $\varphi$ can be captured by ODTA, 
the undecidability of ODTA extended with formulas of the form~(\ref{eq: monotone})
follows immediately.

At this point we would also like to point out that 
extending ODTA with operation such as addition on data values
will immediately yield undecidability.
This can be deduced immediately from~\cite{halpern} 
where we know that together with unary predicates,
addition yields undecidability.

\section{When the Data Values are Strings}
\label{s: string data value}

In this section we discuss data trees where
the data values are strings from $\{0,1\}^{\ast}$, instead of natural numbers.
We call such trees {\em string data trees}.
There are two common kinds of order for strings:
the prefix order, and the lexicographic order.
Strings with lexicographic order are simply linearly ordered domain,
thus, ODTA can be applied directly in such case.

For the prefix order, we have to modify the definition of ODTA.
Consider a string data tree $t$ over the alphabet $\Sigma$.
Let $V_t$ be the set of data values found in $t$.
We define $\V_{\Sigma}(t)$ as a {\em tree} over the alphabet $2^{\Sigma}$,
where
\begin{itemize}
\item
$\Domain(\V_{\Sigma}(t))$ is $V_t \cup \{\epsilon\}$;
\item
for $u,v \in \Domain(\V_{\Sigma}(t))$,
$u$ is a parent of $v$ if $u$ is a prefix of $v$
and there is no $w\in \Domain(\V_{\Sigma}(t))$
such that $u$ is a prefix of $w$ and $w$ is a prefix of $v$;
\item
for $u \in \Domain(\V_{\Sigma}(t))$
the label of $u$ is $S$, if $u \in [S]_t$;
and $\scroot$, if $u = \epsilon$.
\end{itemize}
We call $\V_{\Sigma}(t)$ the {\em tree representation}
of the data values in $t$.
Consider an example of a string data tree in Figure~\ref{fig: string data tree}.
We have 
$$
\begin{array}{lll}
~[\{a\}]_t = \{0101\}  &  & [\{b\}]_t = \{0100\}
\\
~[\{c\}]_t = \{01011\} & & [\{a,b\}]_t = \{01\}
\\
~[\{b,c\}]_t = \{01000\} & & [\{a,b,c\}]_t = \{010011\}.
\end{array}
$$
So $\Domain(\V_{\Sigma}(t)) = \{01,0100,0101,010011,010000,01011\}$,
and 
\begin{itemize}\itemsep=0pt
\item
$01$ is the parent of $0100$ and $0101$;
\item
$0100$ is the parent of $010011$ and $010000$; and
\item
$0101$ is the parent of $01011$.
\end{itemize}

\begin{figure*}
\begin{picture}(300,175)(-150,-120)

\put(-52,40){${a \choose 01}$}
\put(-44,32){\vector(-2,-1){60}}
\put(-44,32){\vector(-3,-4){24}}
\put(-44,32){\vector(3,-4){24}}
\put(-44,32){\vector(2,-1){60}}

\put(-118,-10){${b \choose 0100}$}
\put(-80,-10){${c \choose 01011}$}
\put(-35,-10){${a \choose 010011}$}
\put(5,-10){${a \choose 0101}$}

\put(-69,-18){\vector(-3,-4){24}}
\put(-69,-18){\vector(3,-4){24}}
\put(-69,-18){\vector(2,-1){64}}

\put(-102,-62){${b \choose 01}$}
\put(-60,-62){${b \choose 010011}$}
\put(-15,-62){${a \choose 0101}$}

\put(-93,-70){\vector(0,-1){30}}
\put(-45,-70){\vector(0,-1){30}}
\put(-4,-70){\vector(0,-1){30}}

\put(-110,-110){${c \choose 010011}$}
\put(-62,-110){${c \choose 010000}$}
\put(-20,-110){${b \choose 010000}$}

%%%%%%%%%%%%%%%%%%%%%%%%%%%%%%%%%%%%%%%%%%%%%%%%%%%%%%%%%%%%%%%%%%%%%%%%
%%%%%%%%%%%%%% TREE REPRESENTATION OF DATA VALUES %%%%%%%%%%%%%%%%%%%%%%
%%%%%%%%%%%%%%%%%%%%%%%%%%%%%%%%%%%%%%%%%%%%%%%%%%%%%%%%%%%%%%%%%%%%%%%%

\put(134,40){\small $\scroot$}

\put(145,32){\vector(0,-1){30}}

\put(135,-10){\small $\{a,b\}$}
\put(146,-18){\vector(-3,-4){24}}
\put(146,-18){\vector(3,-4){24}}

\put(115,-62){\small $\{b\}$}
\put(165,-62){\small $\{a\}$}

\put(121,-70){\vector(-3,-4){24}}
\put(121,-70){\vector(3,-4){24}}
\put(173,-70){\vector(0,-1){32}}

\put(80,-110){\small $\{a,b,c\}$}
\put(132,-110){\small $\{b,c\}$}
\put(166,-110){\small $\{c\}$}

\end{picture}

\caption{An example of a string data tree (on the left) and the tree representation of its data values 
(on the right).}
\label{fig: string data tree}
\end{figure*}
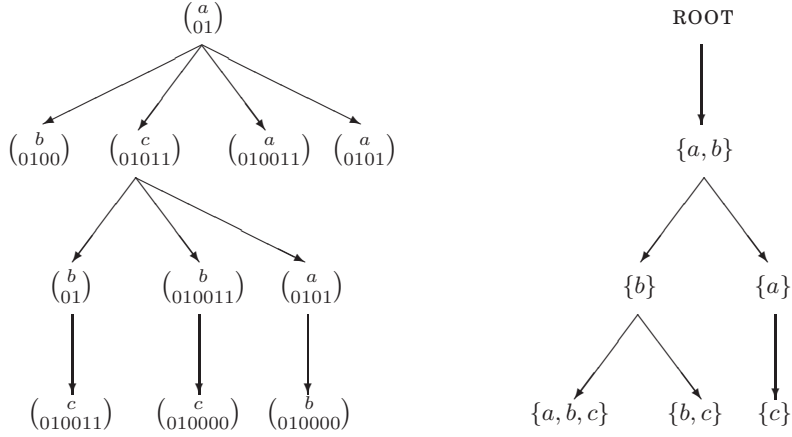

Now an ODTA for string data trees 
is  $\S = \langle \T,\A,\Gamma_0\rangle$, where 
$\T$ is a letter-to-letter transducer from $\Sigma\times\{\top,\bot,*\}^3$
to $\Gamma$;
$\A$ is an unranked tree automaton over the alphabet $2^{\Gamma}$;
$\Gamma_0 \subseteq \Gamma$.
The requirement for acceptance is the same as in 
Section~\ref{s: S-automata},
except that $\A$ takes a tree over the alphabet $2^{\Gamma}$ as the input.

We observe that in the proof of the decidability of
the non-emptiness of ODTA $\S= \langle \T,\M,\Gamma_0\rangle$,
the automaton $\M$ is converted in polynomial time into a Presburger formula
by applying Proposition~\ref{p: aut-to-presburger},
which actually holds for tree automata.
Hence, the decision procedures in Sections~\ref{s: S-automata} and~\ref{s: weak S-automata}
can also be applied to string data trees.

\section{Concluding Remarks}
\label{s: conclusion}

In this paper we study data trees in which the data values come from a linearly ordered domain, 
where in addition to equality test, 
we can test whether the data value in one node is greater than the other. 
We introduce ordered-data tree automata (ODTA), provide its logical characterisation,
and prove that its non-emptiness problem is decidable. 
We also show the logic $\EMSO^2(\eda,\era,\sim)$
can be captured by ODTA.

Then we define weak ODTA,
which essentially are ODTA without the ability
to perform equality test on data values on two adjacent nodes. 
We provide its logical characterisation. 
We show that a number of existing formalisms and models studied in the literature so far 
can be captured already by weak ODTA.
We also show that the definition of ODTA can be easily modified,
to the case where the data values come from a partially ordered domain, such as strings.

We believe that the notion of ODTA provides
new techniques to reason about ordered-data values on unranked trees,
and thus, can find potential applications in practice.
We also prove that ODTA capture
various formalisms on data trees studied so far in the literature.
As far as we know this is the first formalism for data trees
with neat logical and automata characterisations.

\vspace{0.5 cm}
\noindent
{\bf Acknowledgement.}
The author would like to thank FWO for their 
generous financial support under the scheme
FWO Marie Curie fellowship. 
The author also thanks Egor V.~Kostylev for careful proof reading of this paper,
as well as Nadime Francis for pointing out the reference~\cite{halpern}.
Finally, the author thanks the anonymous referees, 
both the conference and the journal versions,
for their careful reading and comments which greatly improve the paper.

\vspace{0.5 cm}

\bibliography{tocl-aut-fo2--2012-bib}
\bibliographystyle{acmtrans}

%\input{appendix-reach.tex}

%\begin{received}
%...
%\end{received}

\end{document}